\documentclass[journal,twoside]{IEEEtran}

\usepackage[T1]{fontenc}
\usepackage{amsfonts,amssymb}
\usepackage{amsthm}
\usepackage{mathtools,multicol}
\usepackage{array}
\usepackage{soul}
\usepackage[caption=false,font=normalsize,labelfont=sf,textfont=sf]{subfig}
\usepackage{booktabs}
\usepackage{cite}
\newtheorem{theorem}{Theorem}
\newtheorem{example}{Example}

\usepackage{soul} 

\usepackage{tikz}  						
\usetikzlibrary{shapes.geometric, arrows,calc}
\usetikzlibrary{positioning}
\usetikzlibrary{decorations.pathreplacing}
\tikzstyle{arrow} =[thick,->,>=stealth] 

\usepackage{pgfplots}
\pgfplotsset{compat=newest}
\usetikzlibrary{plotmarks}
\usetikzlibrary{arrows.meta}
\usepgfplotslibrary{patchplots}
\usepackage{grffile}
\usepackage{graphicx}
\usepgfplotslibrary{fillbetween}

\usepackage[ruled,vlined]{algorithm2e}

\definecolor{Cblue}{RGB}{31 119 180}
\definecolor{Cred}{RGB}{214 39 40}
\definecolor{Corange}{RGB}{255 127 14}
\definecolor{Cgreen}{RGB}{44 160 44}
\definecolor{Cpurple}{RGB}{148 103 198}
\definecolor{Cbrown}{RGB}{140 86 75}
\definecolor{Cpink}{RGB}{227 119 194}
\definecolor{Ccyan}{RGB}{23 190 207}
\definecolor{Cgray}{RGB}{127 127 127}
\definecolor{Cyellow}{RGB}{255 204 0}
\definecolor{Cdarkgray}{RGB}{100 100 100}

\newcommand{\xpol}{{\mathsf{x}}}
\newcommand{\ypol}{{\mathsf{y}}}

\newcommand{\vect}[1]{\ensuremath{\boldsymbol{ #1 }}}

\usepackage{tabstackengine}
\setstackEOL{\cr}

\usepackage[colorlinks=true,bookmarks=false,citecolor=blue,urlcolor=blue]{hyperref} 
\usepackage{tikz,xcolor,hyperref}

\definecolor{lime}{HTML}{A6CE39}
\DeclareRobustCommand{\orcidicon}{
	\includegraphics{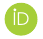}
}
\foreach \x in {A, ..., Z}{\expandafter\xdef\csname orcid\x\endcsname{\noexpand\href{https://orcid.org/\csname orcidauthor\x\endcsname}
			{\noexpand\orcidicon}}
}

\definecolor{C9}{RGB}{183 69 254}
\definecolor{C8}{RGB}{166 80 228}
\definecolor{C7}{RGB}{149 92 202}
\definecolor{C6}{RGB}{131 103 176}
\definecolor{C5}{RGB}{118 111 155}
\definecolor{C4}{RGB}{96 126 122}
\definecolor{C3}{RGB}{79 138 96}
\definecolor{C2}{RGB}{61 149 70}
\definecolor{C1}{RGB}{45 159 45}

\begin{document}
\title{Data-driven Enhancement of the Time-domain First-order Regular Perturbation Model}
\author{Astrid~Barreiro\hspace{-0.1cm}\orcidA{}\hspace{-0.2cm},~\IEEEmembership{Student Member,~IEEE,}
        Gabriele~Liga\hspace{-0.1cm}\orcidB{}\hspace{-0.2cm},~\IEEEmembership{Member,~IEEE,}
        and~Alex~Alvarado\hspace{-0.1cm}\orcidD{}\hspace{-0.2cm},~\IEEEmembership{Senior Member,~IEEE}
\thanks{A. Barreiro, G. Liga, and A. Alvarado are with the Department
of Electrical Engineering, Eindhoven University of Technology, 5600 MB Eindhoven, The Netherlands. e-mail: a.barreiro.berrio@tue.nl.}
\thanks{Parts of this work have been presented at the IEEE Photonics Conference IPC  Vancouver, Canada, Nov. 2022 \cite{Barreiro2022}.}
\thanks{The work of A. Barreiro and A. Alvarado has received funding from the European Research Council (ERC) under the European Union's Horizon 2020 research and innovation programme (grant agreement No 757791). G.~Liga gratefully acknowledges the EuroTechPostdoc programme under the European Union's Horizon 2020 research and innovation programme (Marie Sk\l{}odowska-Curie grant agreement No. 754462).}}

\markboth{Preprint, \today}%
{Barreiro \MakeLowercase{\textit{et al.}}: On a Data-driven Enhancement for the First-order Regular Perturbation Model}

\maketitle

\begin{abstract}
A normalized batch gradient descent optimizer is proposed to improve the first-order regular perturbation coefficients of the Manakov equation, often referred to as kernels. The optimization is based on the linear parameterization offered by the first-order regular perturbation and targets enhanced low-complexity models for the fiber channel. We demonstrate that the optimized model outperforms the analytical counterpart where the kernels are numerically evaluated via their integral form. The enhanced model provides the same accuracy with a reduced number of kernels while operating over an extended power range covering both the nonlinear and highly nonlinear regimes.  A $6-7$~dB gain, depending on the metric used, is obtained with respect to the conventional first-order regular perturbation.
\end{abstract}



\section{Introduction}

 \IEEEPARstart{C}{hannel} models are the cornerstone in the design of fiber-optic communication systems. Modeling provides physical insights into the light propagation phenomena and yields techniques to effectively compensate for \textit{nonlinear interference (NLI)}, arguably the most significant factor limiting the capacity of long-haul coherent optical communication systems \cite[Sec. 9.1]{Bononi2020b}. A channel model in the form of a reasonably simple expression that, given the input to the channel, provides the corresponding output, is essential. Therefore, research on modeling has been a central topic in fiber-optic communications for many years \cite{Chen2010,Poggiolini2011,Carena2014}\footnote{A comprehensive timeline of channel modeling efforts can be found in \cite[Sec. I-A]{Rabbani2020}.}.  
 
 The origin of most analytical models for coherent systems is either the nonlinear Schr\"odinger (NLS) equation or the Manakov equation \cite[Sec. 9.1]{Bononi2020b}. These are the equations governing the signal propagation in fibers, and thus, finding their solution is crucial for predicting the NLI. None of these equations have closed-form solutions for arbitrary transmitted pulses. However, approximated solutions exist in the framework of perturbation theory \cite{Mecozzi2000,Vannucci2002b,Forestieri2005b,Poggiolini2011,DarOPEX13,Carena2014,Oliari2019b}. Perturbation theory can be used to develop fairly compact analytical expressions for computing the NLI under a first-order approximation. One of the most popular approaches uses first-order perturbation on the nonlinear fiber coefficient, which, following \cite{Ghazisaeidi2017}, we refer to in this paper as FRP. 
 
\begin{figure*}[t!]
    \centering 
    \resizebox{0.9\textwidth}{!}{\includegraphics{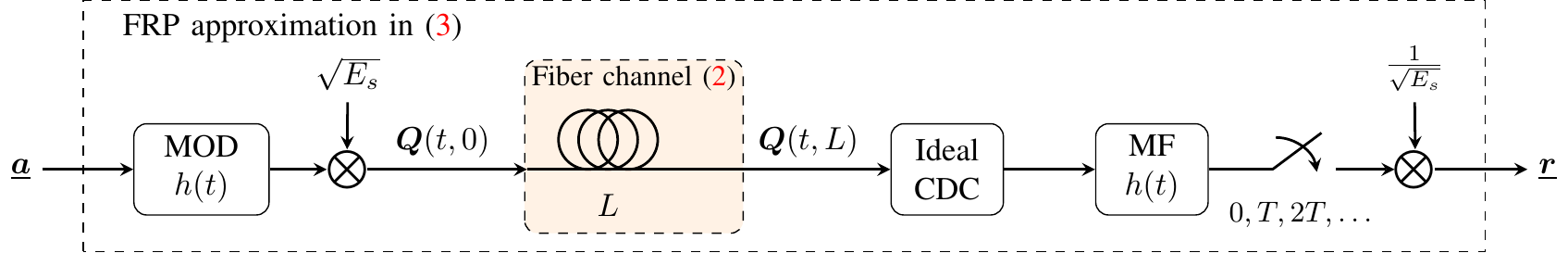}}
    \caption{System model under consideration in this work. The transmitter uses linear modulation (MOD) to generate the transmit signal using the normalized pulse shape $h(t)$. The receiver applies ideal chromatic dispersion compensation (CDC), matched filtering (MF), and sampling. The coefficient $\sqrt{E_s}$ is used to tune the launch power. The optical channel considers a fiber of length $L$, and propagation of the transmit signal through the fiber is given by \eqref{eq:norm_manakov}.} 
        \label{fig:system_model}
\end{figure*}

 FRP has been broadly employed in fiber-optic transmission systems, both in the frequency and time domains. In this paper, we focus on time-domain FRP since it is well suited to design coded modulation systems tailored to the fiber channel. In addition, time-domain FRP has proved great potential to be used in the performance assessment of systems operating in the pseudo linear regime\footnote{A comprehensive summary can be found in \cite[Chap. 9.4]{Bononi2020b}.} \cite{Poggiolini2011,Carena2014}, and for algorithm design in the context of NLI compensation \cite{Tao2011,Liang2014,Rafique2015,Dar17JLT}. 
 
 Although widely employed, time-domain FRP has two main drawbacks. First, as discussed in \cite{Tao2011}, FRP requires the computation of a generally large number of nonlinear perturbation coefficients, which are typically referred to as \emph{kernels}. In particular, to maintain a certain accuracy, the number of kernels that need to be computed grows cubically as $(2N+1)^3$, with $N$ being the channel memory. The effective memory of the channel increases with increments in bandwidth and fiber length \cite[Sec. II-B]{Agrell14}. Thus, kernels' evaluation becomes computationally demanding as transmission bandwidth and fiber length increase. This first drawback can restrict the usage of FRP to transmission scenarios limited in bandwidth and distance.

 The second drawback of time-domain FRP is the loss of precision at high powers. Since the nonlinear contribution in the Manakov/NLSE equation can no longer be considered a perturbation at high powers, higher-order terms in the perturbative expansion become significant. Such behavior sets a power threshold up to which FRP can accurately predict the channel output. This drawback can be avoided by considering higher order terms in the regular perturbation expansion  \cite{Vannucci2002b,Orappanpara2021}. However, such an approach makes the solution analytically more complex, which renders the FRP solution preferable in practice despite its reduced accuracy at high powers.
 
 In the context of the design of nonlinearity compensation/mitigation algorithms, multiple researchers have targeted a reduction of the FRP computational complexity. For instance, the pulse shape was designed to simplify the kernels' computation in \cite{Ghazisaeidi2014,Frey2018}. Other approaches, e.g., \cite{Tao2011,Dar17JLT}, reduce the number of kernels by considering the temporal phase-matching symmetry. Such symmetry enables the pruning of coefficients having zero contribution due to the isotropic phase distribution of the transmitted symbols \cite[Sec. VIII]{Mecozzi2012c}. Other works impose a quantization of the kernels \cite{Peng2015,Malekiha2016,Sorokina2017}. The quantization procedure equates sets of perturbation coefficients to a single value, leading to a significant reduction in the number of kernels to be computed for the FRP approximation. More recently, a significant number of data-driven approaches empowered by machine learning algorithms have been reported to perform the optimization of equivalent FRP coefficients \cite{Zhang2019,Gao2019,Melek2020,Redyuk2020} to be used in nonlinearity compensation algorithms.
 
 In this work, we propose a data-driven optimization of FRP-like kernels to generate an equivalent numerical FRP model addressing the two main drawbacks of the analytical FRP model. Our enhanced model provides the same accuracy with a reduced number of kernels while operating over an extended power range. The main contributions of our work compared to previous works are i) the kernels' optimization in this paper specifically targets improved low-complexity models for the fiber channel; ii) the optimization makes use of the linear parameterization offered by the FRP formalism. Optimizing the kernels in the FRP formalism reduces the modeling complexity and provides useful insights into some channel properties such as its effective memory length. The numerical results show that optimized kernels yield an effective FRP model that for the system under consideration extends the validity region of FRP 6-7~dB above the pseudo-linear threshold, significantly improves the model matching in phase and magnitude to the true value, and generates a memory reduction at a fixed model precision that translates into a complexity reduction of the model computation. 
 
 The paper is organized as follows. In Sec. \ref{sec:system_model} we introduce the transmission system model and briefly review the essentials of the time-domain FRP, also reported in \cite{DarJLT2015}. In Sec. \ref{sec:frp_performance} we assess the FRP performance that is the baseline for the subsequent analysis. In Sec. \ref{sec:enhanced_FRP} we present the essentials of our optimizer, where some examples are given to illustrate the vectorization of gradient descent. In Sec. \ref{sec:nbgd_results} we discuss the performance of the optimized model and discuss the implications over the validity region of FRP and its complexity. Section VI is devoted to conclusions. 

\section{Transmission system model}\label{sec:system_model}
Throughout this paper, a dual-polarization single-span unrepeated fiber-optic transmission system is considered. The block diagram in Fig.~\ref{fig:system_model} illustrates the system model under study. As the purpose of this model is to mainly study NLI, amplified spontaneous emission (ASE) noise is not taken into account.  Furthermore, we only consider single-channel transmission for simplicity of illustration of the proposed enhanced model.

 First, a sequence\footnote{\textbf{Notation convention:} We use boldface letters to denote column vectors, e.g., $\vect{u}$. Underlined bold letters represent infinite sequences of vectors, e.g., $\underline{\vect{u}}$. $|\cdot|$ denotes absolute value. When $|\cdot|$ is applied to a set, it denotes cardinality.  For any pair of vectors $\vect{u}$ and $\vect{v}$, we use $\oslash$ to denote the element-wise division, i.e., $\vect{w} =\vect{u}\oslash \vect{v}$ implies $w_i = u_i/v_i$. The operations $(\cdot)^\mathrm{T}$ and $(\cdot)^{\dagger}$ are the transpose and the Hermitian transpose respectively, and calligraphic letters are used to denote sets. $\mathbb{Z}$ denotes the set of integers, while $\mathbb{C}$ is the set of complex numbers. Throughout this paper, we often use triple indexation (e.g., $G_{klm}$), which we sometimes write using separating commas (e.g., $G_{k,l,m}$).} of two-dimensional complex symbols $\underline{\vect{a}}=\ldots, \vect{a}_{n-1},\vect{a}_{n},\vect{a}_{n+1},\ldots$ is used to linearly modulate the energy-normalized real pulse shape $h(t)$, i.e., $\int_{-\infty}^\infty h^2(t)\text{d}t = 1$. The transmit signal $\vect{Q}(t,0)$  at location $z=0$ is given by
 
 \begin{align}
     \label{eq:lin_mod}
     \vect{Q}(t,0) = \sqrt{E_s} \sum_{n=-\infty}^{\infty} \vect{a}_{n} h(t-nT),
 \end{align}
 where $T$ is the symbol duration, and $E_s$ is the average energy per transmitted symbol. Assuming the symbols $\vect{a}_{n}$ are taken form a normalized constellation, the parameter $E_s= P/2R_s$ in \eqref{eq:lin_mod} defines the total transmitted optical power, where $P$ and $R_s=1/T$ are the launched power and the symbol rate, respectively.

 Each symbol $\vect{a}_{n}$ in the sequence is $\vect{a}_{n} = (a_{\xpol,n}, a_{\ypol,n})^\mathrm{T}$, where $a_{\xpol,n}$ and $a_{\ypol,n}$ represent the complex symbols mapped onto two arbitrary orthogonal polarization states $\xpol$ and $\ypol$. The noiseless propagation of the two-dimensional complex envelope $\vect{Q}(z,t)$ through the fiber link is governed by the Manakov equation \cite[eq.~(57)]{Wai96}
 \begin{align}\label{eq:norm_manakov}
 \frac{\partial \vect{Q}}{\partial z} = -\jmath\frac{\beta_2}{2} \frac{\partial^2 \vect{Q}}{\partial t} +\jmath \gamma \frac{8}{9} e^{-\alpha z} |\vect{Q}|^2\vect{Q},
 \end{align}
 where for simplicity $(z,t)$ is omitted. In \eqref{eq:norm_manakov}, $\alpha$ is the attenuation coefficient, $\beta_2$ the group-velocity dispersion coefficient, and $\gamma$ the nonlinear coefficient. The field $\vect{Q}$ is considered to be attenuation-normalized \cite[eq. (3.1.3)]{Agrawal}. The first contribution on the RHS of \eqref{eq:norm_manakov} is associated with linear propagation, while the second is accounting for nonlinear propagation.
 
 At the receiver side in Fig.~\ref{fig:system_model}, ideal chromatic dispersion compensation is performed on the propagated field $\vect{Q}(t,L)$. The resulting field is then matched-filtered and sampled at the symbol rate. Throughout this work, we assume that the matched filter is matched to the transmitted pulse $h(t)$. The sequence of received symbols $\underline{\vect{r}}=\ldots, \vect{r}_{n-1},\vect{r}_{n},\vect{r}_{n+1},\ldots$ is obtained by scaling the samples by $1/\sqrt{E_s}$. Each output symbol $\vect{r}_{n}$ in the sequence  contains two orthogonal polarization states, namely $\vect{r}_{n} = [r_{\xpol,n}, r_{\ypol,n}]^T$. Average phase rotations on the received constellations are in practice compensated by DSP algorithms. In this paper, we do not consider any additional DSP step beyond what is included in Fig.~(1) (CDC and MF) to set as a benchmark a system that is energy preserving. Based on this choice, we are interested in benchmarking the accuracy of FRP (and enhanced versions thereof) with respect to the SSFM.
 
 For small enough values of the nonlinear coefficient $\gamma$, FRP approximates the exact solution to the Manakov equation yielding the following input-output relation in discrete-time \cite[eq. (3)]{DarJLT2015}:
\begin{align}\label{eq:single_ch_dual_pol_DRP}
	\vect{r}_n\approx \vect{a}_n + \jmath \frac{8}{9} \gamma  E_s \sum_{(k,l,m)\in \mathbb{Z}^3} \left(\vect{a}_{n+k}^\dagger \vect{a}_{n+l}\right) \vect{a}_{n+m} \hspace{0.1cm} S_{klm}.
\end{align}
 In \eqref{eq:single_ch_dual_pol_DRP}, $S_{klm}$ are complex perturbation coefficients that model self-phase modulation (SPM). They are defined as \cite[eq. (4)]{DarJLT2015}
 \begin{align}
 \begin{split}
	S_{klm} \triangleq  &\int_0^L e^{-\alpha z} \int_{-\infty}^{\infty} \hspace{0.1cm} h^*(z,t) h^*(z,t-kT) \\ & h(z,t-lT)h(z,t-mT)\hspace{0.1cm}  \mathrm{d}t \mathrm{d}z,
	\label{eq:SMP_kernel}
 \end{split}
 \end{align}
 where with a slight abuse of notation, we used $h(z,t)$ to denote the solution of \eqref{eq:norm_manakov} when $\gamma=0$ and  $\vect{Q}(t,0)=h(t)$.
 
In general, \eqref{eq:SMP_kernel} could be strictly nonzero for all $(k,l,m) \in \mathbb{Z}^3$ (time-unlimited pulses) which would require the computation of an infinite number of kernels to evaluate \eqref{eq:single_ch_dual_pol_DRP}. However, in practice, due to the exponential decay of the kernel magnitude as a function of the 3D index squared magnitude $k^2+l^2+m^2$ \cite[Fig. 5]{Tao2011}, the sums in \eqref{eq:single_ch_dual_pol_DRP} can be truncated with limited loss in accuracy, yielding a finite memory channel model. In this work, we consider the following truncation
 \begin{align}
  \label{eq:single_ch_dual_pol_DRP_finite}
  \vect{r}_n &\approx \vect{a}_n + \Delta \vect{a}_n,
\end{align}
where
\begin{align}
\label{eq:single_ch_dual_pol_DRP_finite2}
  \Delta \vect{a}_n &\triangleq \jmath \frac{8}{9} \gamma E_s \sum_{(k,l,m)\in\mathcal{S}} \left(\vect{a}_{n+k}^\dagger \vect{a}_{n+l}\right)\vect{a}_{n+m}~S_{klm},
 \end{align}
 and 
 \begin{equation}\label{eq:DRP_2}
  \mathcal{S}\triangleq \{(k,l,m)\in\mathbb{Z}^3:-M\leq k,l,m \leq M\}.
 \end{equation}
Since the model in \eqref{eq:single_ch_dual_pol_DRP_finite} resembles the heuristic finite-memory channel model introduced in \cite{Agrell14}, we call it \textit{finite-memory} FRP. Accordingly, $M$ can be interpreted as the \emph{model's memory} size, and $2M+1$ defines the size of the \textit{interfering  window}. The symbols within this window are needed to compute the finite-memory model’s output \eqref{eq:single_ch_dual_pol_DRP_finite}. Henceforth, we refer to finite-memory FRP simply as FRP. 

 \section{Accuracy and limitations of FRP}\label{sec:frp_performance}
In order to study the limitations of the FRP, we first investigate a representative transmission scheme outlined within the 400ZR implementation agreement. In this section, we address the FRP drawbacks from a nonlinearity modeling perspective as opposed to focusing on the compensation approach, as it has been done in for example \cite{Kumar2019}. To the best of our knowledge, a thoughtful performance assessment has not been previously addressed from a modeling perspective. Hence we present a study of the loss of accuracy of FRP at high powers via a precise quantification of the discrepancies with reliable simulations of fiber propagation.

The study case considers a $L=120$~km standard single-mode fiber span, for a dual-polarization transmission in a single $60$~Gbd channel using $16$-QAM. Table \ref{table:simparam} summarizes the considered fiber parameters. The single-span transmission system constitutes the building block for the first stage of a multi-span EDFA system. When a multi-span transmission is considered, we do not expect the results to be qualitatively different from the ones observed in a single-span scenario. In a multi-span scenario, the model's mathematical formalism remains the same and that second-order effects are mainly affected by the transmitted power rather than transmission distance.

Throughout this paper, the standard split-step Fourier method (SSFM) is used for the simulation of fiber propagation according to \eqref{eq:norm_manakov}. The simulations are performed with a sampling rate equal to four times the transmission bandwidth and a uniform step size equal to $10$~m. Although ASE noise is neglected in Fig.~\ref{fig:system_model} and in \eqref{eq:norm_manakov}, in some of the results we consider an erbium-doped fiber amplifier (EDFA). In such cases, a $5$~dB noise figure is assumed. That system is used to set a baseline to study the performance of FRP around and beyond the optimum launch power. In addition, a root-raised-cosine (RRC) is chosen as pulse shape to numerically calculate the nonlinear coefficients in \eqref{eq:SMP_kernel}. Table \ref{table:simparam} summarizes the considered pulse and parameters.
\begin{table}[t!]
\def\arraystretch{0.9}
\caption{Fiber and pulse shape parameters}
\centering
{\footnotesize
\begin{tabular}{c c}
\toprule 
Nonlinear parameter $\gamma$        & $1.2 \: \textrm{W}^{-1} \textrm{km}^{-1}$ \\
Fiber attenuation $\alpha$          & $0.2 \: \textrm{dB/km}$ \\
Group velocity dispersion $\beta_2$ & $-21.7 \: \textrm{ps\textsuperscript{2}/km}$ \\
Pulse shape $h(t)$                  & Root-raised-cosine (RRC) \\
RRC roll-off factor                 & 0.01 \\
\bottomrule
\end{tabular}}
\label{table:simparam} 
\end{table}

In what follows, we introduce and discuss the results of four metrics used in this work to assess the performance of FRP. The first three metrics are average metrics, while the last one is a point-wise metric. Our general intention in this section is to characterize the well-known FRP drawbacks to set a baseline for the analysis of the results following the model optimization.

\subsection{Signal-to-noise ratio (SNR)}\label{sec:snr_def}
 Let $A_{\xpol/\ypol}$ and $R_{\xpol/\ypol}$ be complex random variables corresponding to the $\xpol/\ypol$ component of the transmitted and received symbols in Fig.~\ref{fig:system_model}, respectively. Throughout this paper we assume that the transmitted symbols are drawn from a polarization-multiplexed format, and thus, the 4D complex symbols are the Cartesian product of a constituent 2D complex constellation by itself. The support of the random variable $A_{\xpol/\ypol}$ is the constellation  $\mathcal{A}\subset\mathbb{C}$, given by $|\mathcal{A}|$ constellation points $\mathcal{A}\triangleq\{s_1,s_2,\ldots,s_{|\mathcal{A}|}\}$. 
 
 Let $\mu\in \mathbb{C}$ and $\sigma^2\in \mathbb{R}$ be two functions of the random variable $A_{\xpol/\ypol}$ representing the conditional mean and conditional variance of the general constellation point $A_{\xpol/\ypol}$, respectively. These quantities are defined as
 \begin{align}
 \label{eq:mu}
\mu (A_{\xpol/\ypol})       & \triangleq \mathbb{E} \{R_{\xpol/\ypol}|A_{\xpol/\ypol}\}, \\
\label{eq:sigma}
\sigma^2(A_{\xpol/\ypol})    & \triangleq \mathbb{E} \{|R_{\xpol/\ypol}-\mu(A_{\xpol/\ypol})|^2|A_{\xpol/\ypol}\}.
\end{align}
 We assume that $A_{\xpol/\ypol}$ is zero mean, and unit energy ($\mathbb{E}\{|A_{\xpol/\ypol}|^2\}=1$). Additionally, $A_{\xpol}$ and $A_{\ypol}$ are assumed independent. For reasons that will become clear in Sec. \ref{subsec:radii_phase}, it is assumed that $\mathcal{A}$ does not include $0 + \jmath 0$ and that it is a constellation with more than one symbol per ring. 

Using \eqref{eq:mu} and \eqref{eq:sigma}, we define the signal-to-noise ratio (SNR) as the average SNRs across the two polarizations, i.e.,
 \begin{align}
     \mathrm{SNR} & \triangleq \frac{1}{2} \left(\frac{\mathbb{E}\{|\mu(A_\xpol)|^2\}}{\mathbb{E}\{\sigma^2(A_\xpol)\}} + \frac{\mathbb{E}\{|\mu(A_\ypol)|^2\}}{\mathbb{E}\{\sigma^2(A_\ypol)\}}\right).
 \label{eq:snr_def}    
 \end{align}
 
\begin{figure}[!t]
    \centering 
    \resizebox{0.5\textwidth}{!}{\includegraphics{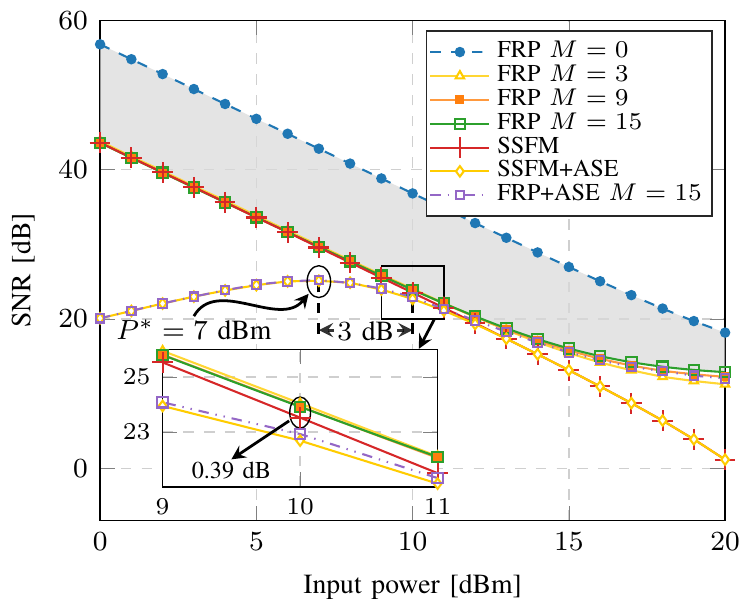}}
    \caption{The SNR in \eqref{eq:snr_def} as a function of launch power $P$ with (``+ASE'') and without ASE noise. Increments in memory size close the gap between SSFM and FRP in the absence of ASE noise for powers up to 10~dBm. The optimum launch power for the ``+ASE'' case $P^*=7$~dBm is shown.~The gray area delimits the region covered for $ M \in [0,15]$. The FRP prediction (with and without ASE) starts to fail at about $3$~dBm above $P^*$.} \label{fig:snr_vs_p}
\end{figure}

 Fig.~\ref{fig:snr_vs_p} shows the SNR in \eqref{eq:snr_def} obtained using SSFM and FRP for different input powers $P$ and three model-memory sizes. The gray area depicts the region covered by FRP with memory size $M \in[0,15]$. Fig.~\ref{fig:snr_vs_p} also includes results with ASE noise, for which $P^*=7$~dBm is found to be the optimum launch power.

 Fig.~\ref{fig:snr_vs_p} shows that already for $M> 1$, the SNR prediction of FRP matches SSFM within 0.5~dB (0.39 dB) in the linear and pseudo-linear regimes. This good fit is lost when $P>10$~dBm, a region where the nonlinear distortions are large. For powers above $10$~dBm, the SNR curves of FRP and SSFM begin to diverge. Therefore FRP's accuracy extends up to $3$~dB above the optimum launch power. The divergence observed above $10$~dBm implies that FRP is (a) underestimating the NLI, (b) it is making an inaccurate prediction of the conditional means, or (c) is doing both, (a) and (b), simultaneously. In the highly nonlinear regime (i.e., $P>15$~dBm), the FRP model shows a saturation trend that is as more visible as $M$ increases. This saturating behavior will be discussed in Sec.~\ref{subsec:radii_phase}.

 To understand the divergent curves in Fig.~\ref{fig:snr_vs_p}, a more qualitative comparison between FRP and SSFM is displayed in Fig.~\ref{fig:clouds}. Three scenarios are considered: 10~dBm for $M=5$ (a), 13~dBm for $M=5$ and $M=15$, (b) and (c) respectively. In Fig.~\ref{fig:clouds}~(a), a mismatch between the FRP (purple) and SSFM (red) is already noticeable even though the memory size considered ($M=5$) is relatively large. The constellation clouds, in this case, appear to have still similar average variance, but in the SSFM case show an \textit{extra} phase rotation not accounted for by FRP. In Fig.~\ref{fig:clouds}~(b) we show the comparison again for $M=5$ but at $P=13$~dBm ($6$~dB above $P^*$). The mismatch between the constellations in Fig.~\ref{fig:clouds}~(b) worsens. We note that the FRP clouds here are not simply rotated compared to SSFM, but also scaled (up). In Fig.~\ref{fig:clouds}~(b) we illustrate the rings where the conditional means of SSFM and FRP fall, and it is visible they do not overlap. Notice that the rings do not match and that the FRP ring has a radius larger than SSFM. Therefore, we conclude that FRP is making an inaccurate prediction of conditional means. Lastly, Fig.~\ref{fig:clouds}~(c) shows results at $P=13$~dBm when the model's memory size is increased from $M=5$ to $M=15$. Despite the increment in memory size, the mismatch persists and the radii discrepancy between SSFM and FRP enlarges (a characteristic further explored in Sec.~\ref{subsec:radii_phase}). This shows that the model's proximity to SSFM quickly saturates as a function of memory size. We further investigate the constellation mismatch based on the three metrics we introduce.
\begin{figure*}[!t]
\centering
\resizebox{0.32\textwidth}{!}{\includegraphics{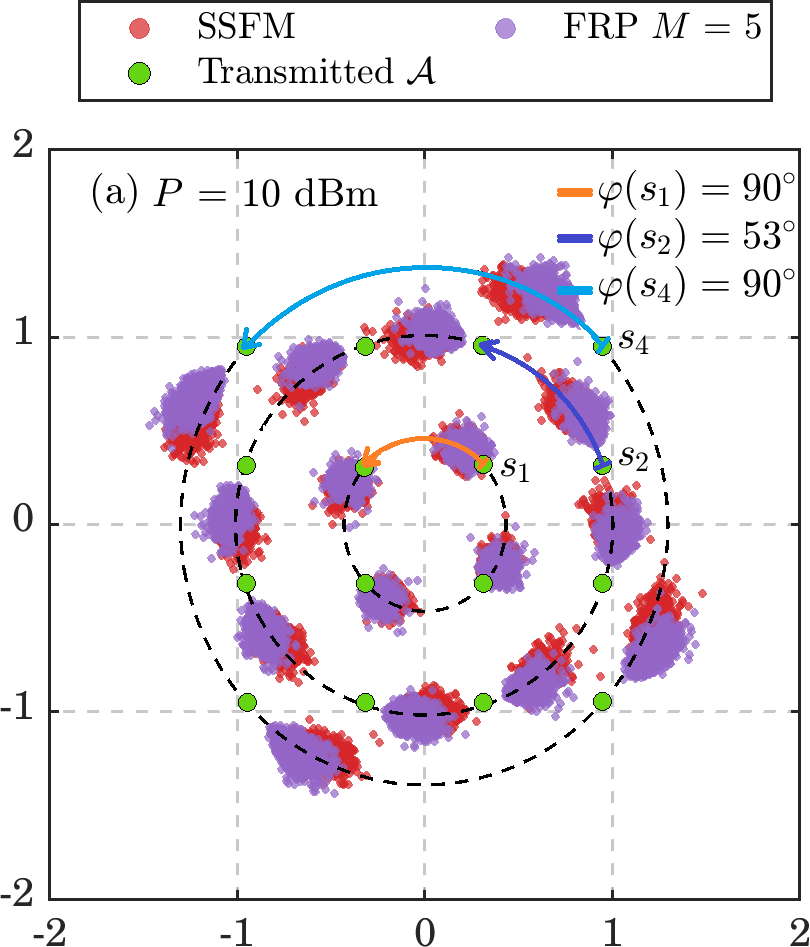}}
\hfill
\resizebox{0.32\textwidth}{!}{\includegraphics{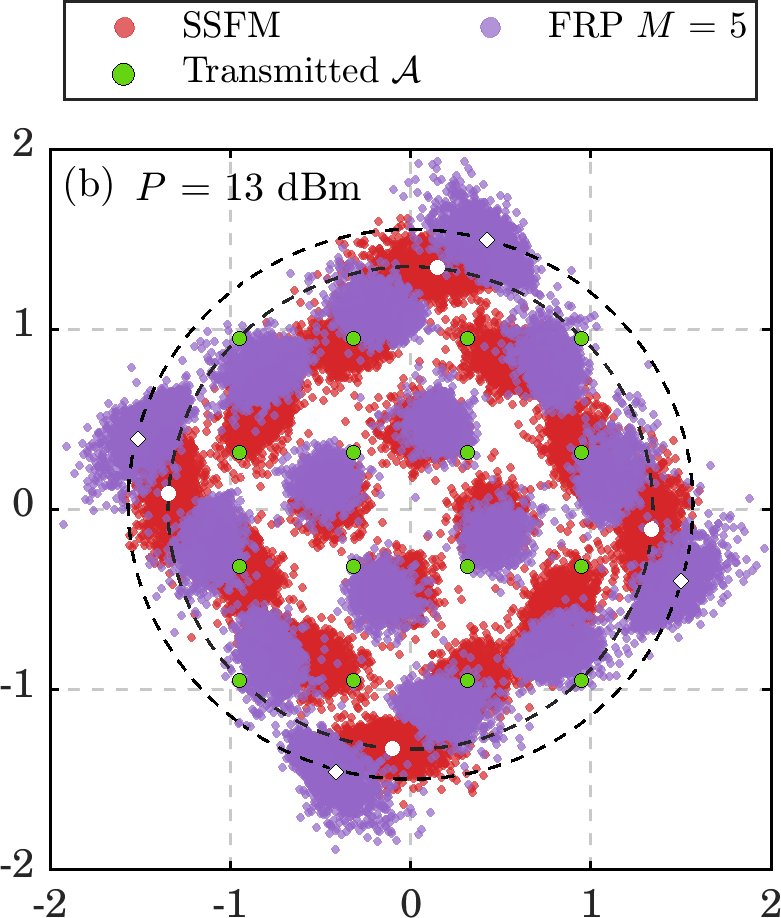}}
\hfill
\resizebox{0.32\textwidth}{!}{\includegraphics{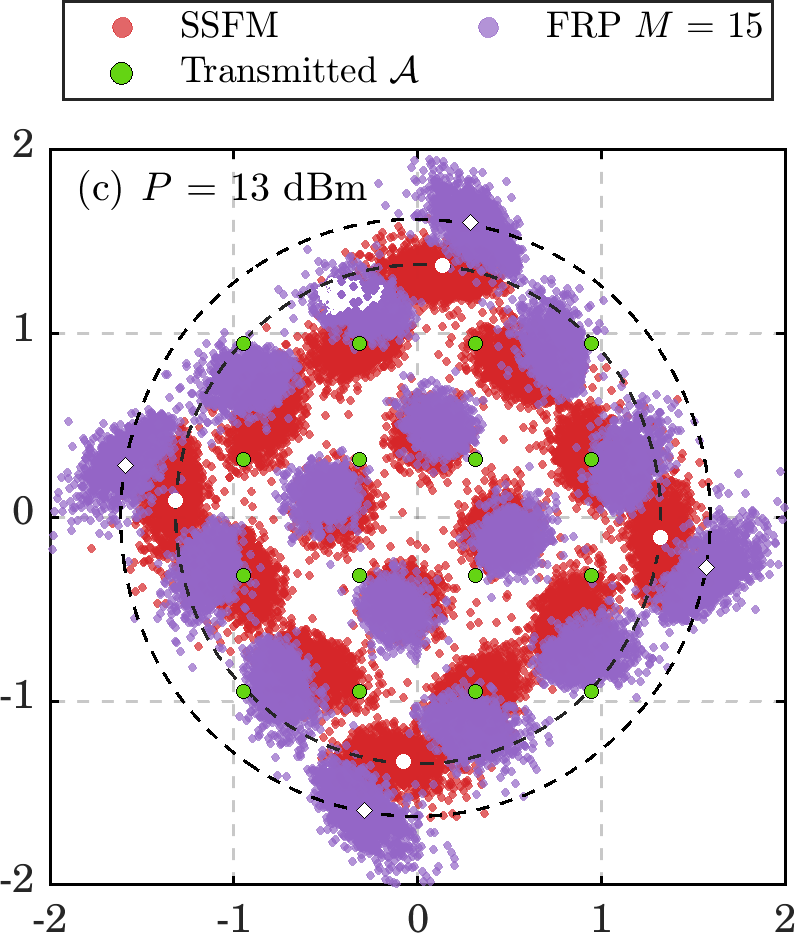}}
\caption{Received constellation diagrams for (a) $P=10$~dBm and FRP $M=5$, (b) $P=13$~dBm and FRP  $M=5$, and (c) $P=13$~dBm and FRP $M=15$. (b) and (c) show explicitly the radii where the SSFM and FRP conditional means fall on. Overall, a mismatch between the SSFM and FRP clouds is observed, most visibly in phase but also in scale. For $P=$13~dBm, the phase mismatch does not improve with memory size increments and the radii difference between SSFM and FRP broadens.} \label{fig:clouds}
\end{figure*}

\subsection{Radii and Phase difference}\label{subsec:radii_phase}
 To better assess the observed mismatch between the received constellations in Fig.~\ref{fig:clouds} and understand the reason for the SNR prediction mismatch shown in Fig.~\ref{fig:snr_vs_p}, we consider two metrics that quantify how different the conditional means ($\mu(s)$ in \eqref{eq:mu}) of SSFM and FRP are with respect to the transmitted constellation points ($\mathcal{A}$). The first metric we introduce is the normalized radii difference, defined as
 \begin{align} \label{eq:radii_ratio}
     \Delta r \triangleq \frac{1}{2}\left(\mathbb{E} \bigg\{ \frac{|\mu(A_\xpol)|-|A_\xpol|}{|A_\xpol|}\bigg\}  + \mathbb{E} \bigg\{ \frac{|\mu(A_\ypol)|-|A_\ypol|}{|A_\ypol|}\bigg\}\right).
 \end{align}
 The normalized radii difference is such that $-1 \leq \Delta r < \infty$. Three cases are of interest. When $\Delta r=0$, the conditional means perfectly match the transmitted symbols. When $\Delta r<0$, the magnitude the of conditional means of the received constellation are on average smaller with respect to the transmitted symbols. In other words, the constellation is ``compressed''. Conversely, when $\Delta r\geq 0$, the received constellation experiences an expansion. 
 
 Secondly, we compare the average phase rotation experienced by the conditional means of SSFM and FRP with respect to transmitted constellation points. The average phase difference is defined as
\begin{align}\label{eq:phase_diff}
    \Delta \phi \triangleq \frac{1}{2}\left(\mathbb{E} \bigg\{\frac{\angle \mu(A_\xpol) - \angle A_\xpol}{\varphi(A_\xpol)}\bigg\}  + \mathbb{E} \bigg\{ \frac{\angle \mu(A_\ypol) - \angle A_\ypol}{\varphi(A_\ypol)}\bigg\}\right),
\end{align}
 where $\angle$ denotes the \textit{angle of}, and $\varphi(A_{\xpol/\ypol})$ is the phase between two neighboring symbols within the same constellation ring, i.e., it is the minimum phase rotation required to transform a symbol $s_m \in \mathcal{A}$ into another one in the same ring. Fig.~\ref{fig:clouds}~(a) shows three instances (orange, blue, and light blue) of $\varphi(A_{\xpol/\ypol})$. The average phase difference $\Delta \phi$ is zero when the conditional means of the received constellation match in phase the transmitted symbols.

\begin{figure*}[!t]
\centering
\resizebox{0.49\textwidth}{!}{\includegraphics{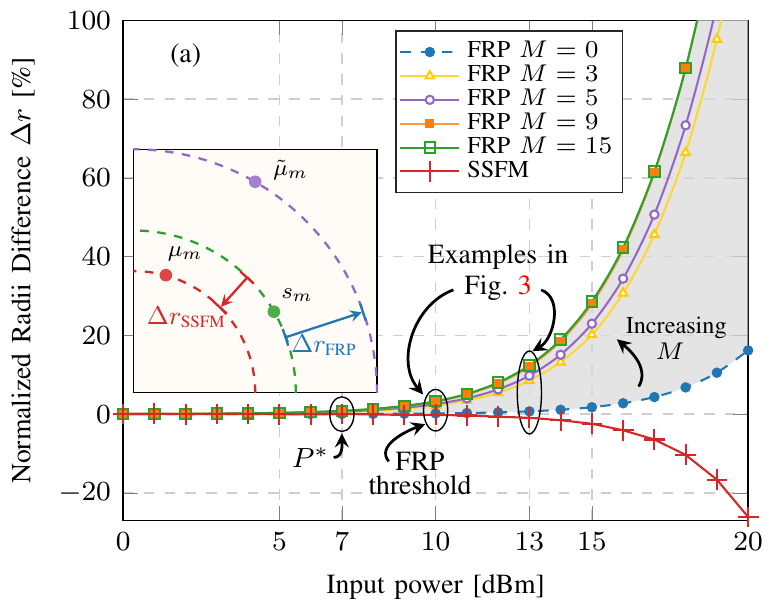}}
\hfill
\resizebox{0.49\textwidth}{!}{\includegraphics{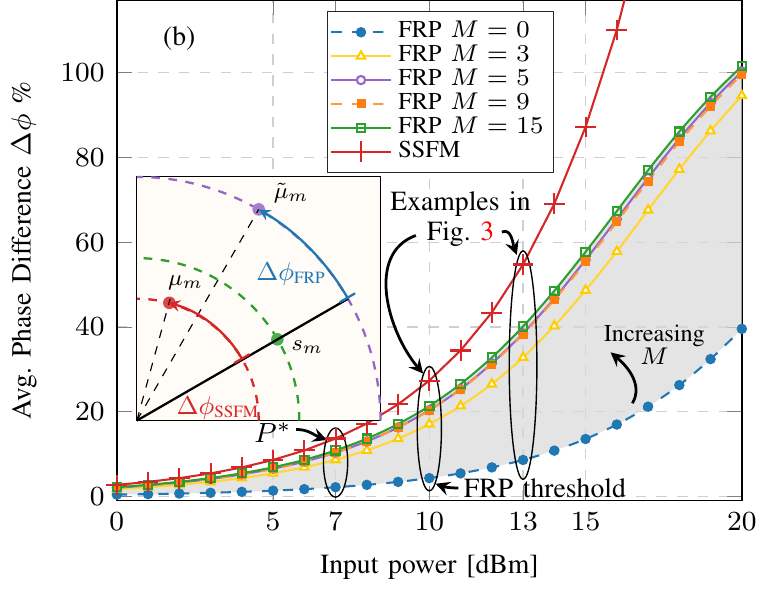}}
\caption{(a) Normalized radii difference as a function of launch power, and (b) phase difference as a function of power. The inset is both figures illustrate how $\Delta r$ and $\Delta \phi$ are defined with respect to a transmitted symbol $s_m$. Dotted lines within the insets depict the radii where the transmitted and received symbols lay. The gray area delimits the region covered for $M \in [0,15]$, and $M =3,5$, and $9$ are shown explicitly. (a) shows FRP moves away from SSFM, while (b) moves towards SSFM.} \label{fig:summary_radii_phase}
\end{figure*}

 Fig.~\ref{fig:summary_radii_phase}~(a) and (b) display the results for the normalized radii difference and average phase difference in \eqref{eq:radii_ratio} and \eqref{eq:phase_diff}, respectively. Illustrative insets have been included for a visualization of how the radii and phase differences are taken for a symbol $s_m$, where we have denoted the expected and model-predicted conditional means as  
 \begin{align}
 \label{eq:mu_SSFM}
    \mu_m &\triangleq  \mu^{\text{SSFM}}(A_{\xpol/\ypol}=s_m) = \mathbb{E} \{R^{\text{SSFM}}_{\xpol/\ypol}|A_{\xpol/\ypol}=s_m\}, \\
 \label{eq:mu_FRP}
    \tilde\mu_m & \triangleq  \mu^{\text{FRP}}(A_{\xpol/\ypol}=s_m)  = \mathbb{E} \{R^{\text{FRP}}_{\xpol/\ypol}|A_{\xpol/\ypol}=s_m\}.
\end{align}
 The tildes in \eqref{eq:mu_FRP} (and more generally throughout this paper) denote approximated values using FRP. The color coding in these insets follows Fig \ref{fig:clouds}. Since we examined $\Delta r$ and $\Delta \phi$ for multiple memory sizes, in Fig.~\ref{fig:summary_radii_phase}~(a) and (b) we highlighted in gray the region spanned by memory sizes $M \in [0,15]$, showing with an arrow the direction the result moves with increments in memory size.

 The results shown in Fig.~\ref{fig:summary_radii_phase}~(a) for $P>P^*$ suggest that, for both SSFM and FRP, the magnitude of the conditional means diverges from the magnitude of their corresponding transmitted constellation points, i.e., $\Delta r \neq 0$. The way $\Delta r$ diverges from zero is, however, not the same for SSFM and FRP. On one hand, the SSFM results in Fig.~\ref{fig:summary_radii_phase}~(a) show a compression  ($\Delta r_{\mathrm{SSFM}}<0$). This compression is due to the energy-preserving nature of the Manakov equation \eqref{eq:norm_manakov} \cite{Kramer2015} and, thus, the energy associated with the increasing NLI variance must be balanced by a reduction in the average energy of the $\mu$ terms. This effect is only significant in the highly nonlinear regime (beyond $13$~dBm). On the other hand, the FRP results in Fig.~\ref{fig:summary_radii_phase}~(a) suggest an expansion ($\Delta r_{\mathrm{FRP}}>0$). Thus, unlike SSFM, FRP is not energy-preserving. The phase difference results in Fig.~\ref{fig:summary_radii_phase}~(b) show that the gap between SSFM and FRP increases with increments in power. Although this gap narrows with increments in memory size, it saturates at about $M=5$. The residual gap is, thus, to be attributed to second-order effects.  
 
 Overall, the results shown in Fig.~\ref{fig:summary_radii_phase} reveal the FRP's inaccuracy for powers beyond the optimum launch power. Above $10$~dBm, the system operates in the nonlinear regime, breaking the FRP hypothesis. Once the $10$~dBm is crossed, FRP yields imprecise predictions of the output signal power and the total nonlinear phase rotation. As shown in Fig.~\ref{fig:clouds} (b)-(c), this result does not exhibit improvement with memory size increments and leads to an inaccurate prediction of SNR at high powers. The following theorem shows why FRP yields imprecise predictions of the average amplitude of the conditional means and total nonlinear phase rotation. 
    \begin{theorem}\label{T:conditional_mean}
     The conditional mean in \eqref{eq:mu_FRP} for the finite-memory FRP model can be expressed as
     \begin{align}\label{eq:cond_mean_frp}
        \begin{split}
            \tilde \mu_m  & =  s_m +  \jmath \frac{8}{9}\gamma E_s \Big[s_m(1+ |s_m|^2) S_{000} \\
            &+\sum_{k\in\mathcal{W}} s_m (2S_{kk0}+ S_{k0k}) + s_m^* \mathbb{E}\{A^2_{\xpol/\ypol}\} S_{0kk}\Big],
        \end{split} 
     \end{align}
     where $\mathcal{W}\triangleq \{i\in[-M,M]\setminus \{0\}\}$.
     \end{theorem}
     \begin{proof}
        See Appendix~\ref{app:condmean}.
     \end{proof}
\begin{figure}[!t]
    \centering 
    \resizebox{0.48\textwidth}{!}{\includegraphics{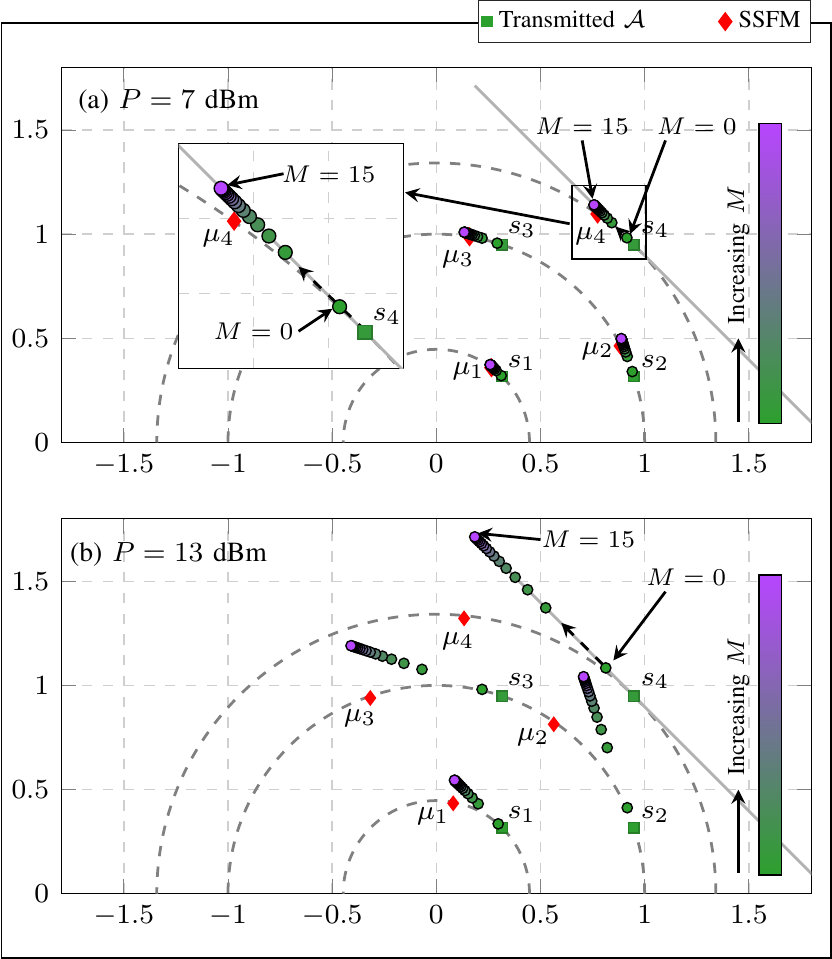}}
    \caption{The evolution with increasing memory $M$ for $\tilde\mu$ in \eqref{eq:cond_mean_frp} is depicted with a color gradient. We consider the four symbols $\{s_1,s_2,s_3,s_4\}$ within the first quadrant of $16$-QAM, and two powers: (a) $P=7$~dBm ($P^*$) and (b) $P=13$~dBm (example Fig.~\ref{fig:clouds} (c)). The rings where $\{s_m\}$ lay are drawn and ${\mu_m}$ (red diamonds)  denote the SSFM conditional means for $s_m$. It can be observed that overall $\{\tilde\mu_m\}$ diverge from the $\{s_m\}$ rings, and they do not match $\{\mu_m\}$ for either (a) nor (b).} \label{fig:mu_trajectory}
\end{figure}
Theorem~\ref{T:conditional_mean} shows that when $E_s \to 0$, then $\tilde \mu_m \to  s_m $, which is expected because $E_s \to 0$ implies that the system operates in the linear regime. More generally, Theorem~\ref{T:conditional_mean} shows that the received symbol has a conditional mean that in general has a different phase and magnitude than the transmitted symbol, see the r.h.s. of \eqref{eq:cond_mean_frp}. In Fig.~\ref{fig:mu_trajectory} we display a numerical evaluation of \eqref{eq:cond_mean_frp} for four representative constellation symbols and multiple values of $M$. Two powers are considered: optimum launch power $P^*$ and $13$~dBm, shown in Fig.~\ref{fig:mu_trajectory} (a) and (b), respectively. It can be observed that for either power case, the FRP conditional means do not converge to the SSFM symbols (red diamonds). This first characteristic is more visible at $13$~dBm than at $7$~dBm. Furthermore, increments in memory do not lead to an improvement. A second noticeable characteristic is that the evolution with memory is defined over a tangent line crossing on $s_m$; and that for high powers, $\mu_m$ are more distant to the tangent. These two characteristics have the following consequences: (i) for all $m$ $\tilde \mu_m>|s_m|$, meaning that $\tilde \mu$ has a different energy than $s_m$; and (ii) increasing the power increases positively the radii difference between the rings where $s_m$ and  $\tilde\mu_m$ are, increasing the difference in energy between $\mu$ and $\tilde\mu$. Therefore, we can conclude that the result in Theorem~\ref{T:conditional_mean} justifies why FRP is non-energy preserving. The results in Fig.~\ref{fig:mu_trajectory} show that FRP is overestimating the received signal power.  

\subsection{Relative error}\label{subsec:RE}
 We further characterize the FRP accuracy by examining the relative error between FRP and SSFM predictions. The two metrics previously studied were based on the behavior of the conditional means, while the relative error better conveys the symbol-wise behavior. The relative error is defined as
 \begin{align} \label{eq:relateverror}
     \varepsilon  \triangleq \sqrt{\frac{1}{2}\left(
     \frac{\mathbb{E}\{|R_\xpol-\tilde{R_\xpol}|^2\}}{\mathbb{E}\{|R_\xpol|^2\}} + \frac{\mathbb{E}\{|R_\ypol-\tilde{R_\ypol}|^2\}}{\mathbb{E}\{|R_\ypol|^2\}}\right)},
 \end{align}
 where $R_{\xpol/\ypol}$ is associated to the the SSFM received symbols in $\xpol/\ypol$ polarization, and $\tilde R_{\xpol/\ypol}$ to the corresponding FRP prediction. $\varepsilon$ is positive-valued and is zero when there is a pointwise match between the model and the SSFM output. 
\begin{figure}[!t]
    \centering 
    \resizebox{0.5\textwidth}{!}{\includegraphics{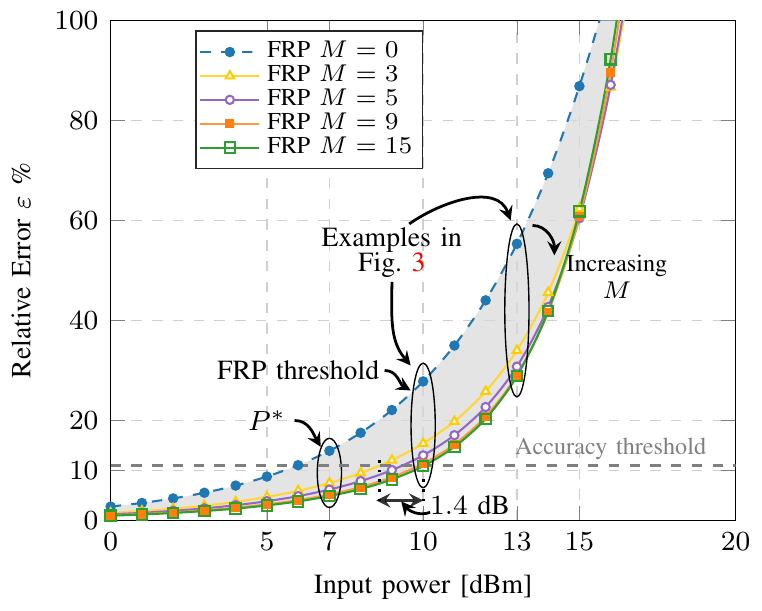}}
    \caption{Relative error of FRP with respect to SSFM in \eqref{eq:relateverror} as a function of power. The gray area delimits the region covered by $M \in [0,15]$. $M =0, 3,5$, $9$, and $15$ are shown explicitly. Overall, increments in memory size lead $\varepsilon$ to decrease. The values of $\varepsilon$ worsen for increments in power, and it reaches very high values for $P>16$~dBm.} \label{fig:relative_error}
\end{figure}

 In Fig.~\ref{fig:relative_error} we plot $\varepsilon$ versus the launch power.  For $P<P^*+3$~dB, $\varepsilon$ in Fig.~\ref{fig:relative_error} takes values below $20$\% for large enough memory sizes ($M\geq3$). At $P=10$~dBm, the FRP model achieves $\varepsilon = 11\%$ for memory size $M=9$ and  $M=15$. From $M=9$ to $M=15$ FRP does not show significant improvement. In this work, we measure the accuracy via $\varepsilon$ and use a threshold of $11$\%, which is the value reached by $M=15$, which is the largest  memory considered in this study. From now on, we shall say that a model is precise if it has $ \varepsilon \leq 11\%$. Fig.~\ref{fig:relative_error} shows that FRP guarantees good accuracy up to $10$~dBm from a memory size of $M=9$. Above $10$~dBm, $\varepsilon$ rapidly increases, and significant increments in memory are demanded to keep the desired accuracy. Consequently, $10$~dBm sets an upper bound for the validity region of FRP in terms of $\varepsilon$. Fig.~\ref{fig:relative_error} shows that above this threshold, increasing $M$ leads to small improvements of $\varepsilon$ with respect to $M=3$, and it begins to saturate at $\approx14$~dBm. It can also be observed in Fig.~\ref{fig:relative_error} that with $M=3$ the accuracy threshold is hit at $P=8$~dBm. It is possible to prolong $1.5$~dB the reached accuracy by increasing the memory from $M=3$ (343 kernels) to $M=15$ (29 791 kernels), meaning that 29 448 kernels more are required to prolong $1.4$~dB the FRP accuracy reached with $M=3$. Notice that for $P=13$~dBm, the example considered in Fig.~\ref{fig:clouds} (b) and (c), the relative error is already quite high ($\approx 30$\%), showing that even though the other two metrics in Fig.~\ref{fig:summary_radii_phase} (a) and (b) are not too pessimistic at this power, $\varepsilon$ is revealing an inaccurate FRP prediction. 

\section{Data-driven estimation of FRP kernels} \label{sec:enhanced_FRP}
 In this section, we study the performance of a data-driven method to optimize the SPM kernels $S_{klm}$. The optimization provides an effective FRP model tailored to extend the model validity beyond the pseudo-linear regime. This data-driven method is an alternative to computing the kernels via their integral form in \eqref{eq:SMP_kernel}. The estimation method relies on a gradient-based optimizer that minimizes the average quadratic error between a set of true-transmission outputs $\vect{r}_1,\cdots, \vect{r}_B$ and the output $  \hat{\vect{r}}_1,\cdots,  \hat{\vect{r}}_B$ determined by the FRP model in \eqref{eq:single_ch_dual_pol_DRP_finite} parameterized with respect to the kernels. In this paper, the true-transmission outputs are generated via SSFM. However, the estimation method can also be used with experimental data.  

Normalized gradient descent (NGD) is a popular enhancement to the standard gradient descent algorithm \cite{Watt}. NGD is specifically designed to ameliorate the vanishing behavior of the magnitude of the negative gradient near stationary points. There are two common approaches to the normalization of the gradient: (i) normalizing the full gradient magnitude, and (ii) normalizing the magnitude component-wise. The latter is considered in this work to define our gradient-based optimizer. In general, the $l$-th NGD descent step is given by
\begin{equation}\label{eq:norm_gd}
    \vect{w}^{(l+1)} = \vect{w}^{(l)} - \alpha \nabla f(\vect{w}^{(l)})\oslash|\nabla f(\vect{w}^{(l)})|,
\end{equation}
where $f: \mathbb{C}^L \mapsto \mathbb{R}$  is a scalar parameterized objective function to be minimized with respect to an $L$-dimensional vector of complex parameters $\vect{w}$. In \eqref{eq:norm_gd} we abuse the notation of the absolute value to refer to an element-wise absolute value. The superscripts $l$ and $l+1$ are integers referring to the current and future NGD stages respectively, and $\alpha$ is the \textit{step} size. If $||\vect{w}^{(l+1)}-\vect{w}^{(l)}||<\tau$, with $\tau$ being a predefined threshold, we say that the NGD has converged. 

To predict the received symbol $\vect{r}_n$ via FRP in \eqref{eq:single_ch_dual_pol_DRP_finite}--\eqref{eq:DRP_2}, the transmitted symbol $\vect{a}_n$, its $2M$ neighbors $\vect{a}_{n-M},\ldots, \vect{a}_{n-1},\vect{a}_{n+1},\ldots,\vect{a}_{n+M}$, and the SPM kernels $S_{klm}$ are required. Conversely, when a transmission pair $(\vect{a}_n,\vect{r}_n)$, and the neighboring symbols of $\vect{a}_n$ are known, then \eqref{eq:single_ch_dual_pol_DRP_finite} yields two linear equations, one per polarization, for which the SPM kernels are \textit{unknowns}. The number of unknowns given a memory size $M$ is 
\begin{equation}\label{eq:numk}
L=(2M+1)^3.   
\end{equation}
Given that two equations are insufficient to determine $L$ kernels, a \textit{batch} of these linear equations must be considered. Our optimizer operates over a batch of data and uses the normalization introduced in \eqref{eq:norm_gd}. Thus, from this point on, we call it normalized batch gradient descent (NBGD). 
\begin{figure}[!t]
    \centering 
    \resizebox{8.5cm}{!}{\includegraphics{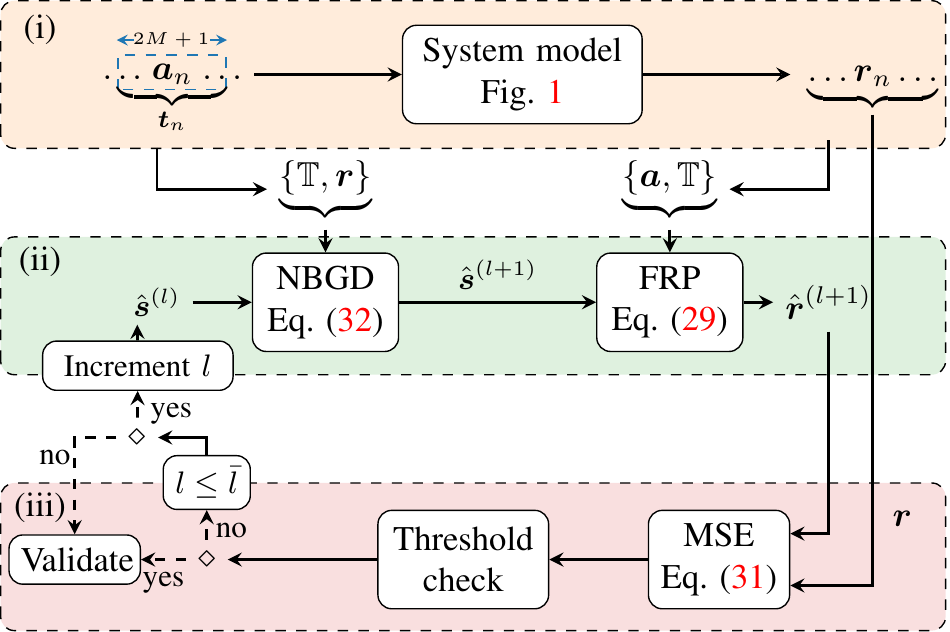}}
    \caption{The NBGD block diagram represents three stages for a single sample, (i) data collection, (ii) gradient step update, and (iii) gradient convergence assessment.} 
        \label{fig:NBGD_sketch}
\end{figure}

Fig.~\ref{fig:NBGD_sketch} shows a summary of our NBGD estimation method. The block diagram describes three stages: (i) data collection, (ii) optimizer update, and (iii) convergence assessment. In the data collection stage, transmission data is generated according to the system model in Fig.~\ref{fig:system_model} using SSFM, as previously indicated in Sec~\ref{sec:frp_performance}. A sliding window is used in the transmitted sequence $\underline{\vect{a}}$ to select a batch of $B$ symbols $\vect{a}_n$ and their $2M$ corresponding neighbors. The window has a size of $2M+1$, with $M$ being the FRP model's memory size. We refer to $B$ as the \textit{batch size}. Similarly, a sliding window is used in the received sequence $\underline{\vect{r}}$ to select an output batch of $B$ symbols $\vect{r}_n$.

The $\xpol$ and $\ypol$ components of \eqref{eq:single_ch_dual_pol_DRP_finite} are given by
\begin{align}
    r_{\xpol/\ypol,n} & \approx a_{\xpol/\ypol,n} +\jmath \frac{8}{9} \gamma E_s \sum_{(k,l,m) \in \mathcal{S}}T_{\xpol/\ypol,klm} \cdot S_{klm},
\label{eq:DRP_finite_x_y}
\end{align}
where 
\begin{align}\label{eq:Ttriplets}
    \begin{split}
        T_{\xpol,klm} & = (a_{\xpol,n+k}^* a_{\xpol,n+l}+a_{\ypol,n+k}^* a_{\ypol,n+l})a_{\xpol,n+m},\\
        T_{\ypol,klm} & = (a_{\ypol,n+k}^* a_{\ypol,n+l}+a_{\xpol,n+k}^* a_{\xpol,n+l})a_{\ypol,n+m}.
    \end{split}
\end{align}

From this point on we call \eqref{eq:Ttriplets} \emph{triplets}. Knowing the triplets and transmission pair $(\vect{a}_n,\vect{r}_n)$ in \eqref{eq:DRP_finite_x_y}, a total of $L$ kernels must be optimized. To do so, either ($\xpol$ or $\ypol$) or both polarizations ($\xpol$ and $\ypol$), can be used. In practice considering both polarization would be the most efficient way to make use of the transmission resources. However, to simplify the notation, we describe here the NBGD optimizer in terms of a single polarization. Thus, the subscript $\xpol/\ypol$ will be omitted from this point on.

By shaping the kernels $S_{klm}$ to be an $L\times 1$ vector $\vect{s}$, and generating an $L \times 1$ vector of triplets $\vect{t}_{n}$ corresponding to $a_{n}$, we vectorize an arbitrary polarization component of \eqref{eq:DRP_finite_x_y} as
\begin{equation}\label{eq:vector_DRP}
    r_{n} \approx a_{n} +\jmath \frac{8}{9} \gamma E_s \vect{t}_{n}^\mathrm{T} \vect{s},   
\end{equation}
 where $\vect{t}_{n}$ and $\vect{s}$  are respectively
    \begin{align}\label{eq:trip_vec}
        \vect{t}_{n}^\mathrm{T}  \triangleq & 
         \bigl( T_{-M,-M,-M},T_{-M,-M,-M+1},T_{-M,-M,-M+2}, \cdots \\ \nonumber
         &T_{-M,-M,M},T_{-M,-M+1,-M}, \cdots T_{M,M,M}\bigr),
    \end{align}
    and 
    \begin{align} \label{eq:kernels_vec}
         \vect{s}^\mathrm{T} \triangleq &        
         \bigl( S_{-M,-M,-M}, S_{-M,-M,-M+1}, S_{-M,-M,-M+2}, \cdots \\ \nonumber
         & S_{-M,-M,M}, S_{-M,-M+1,-M}, \cdots S_{M,M,M}\bigr).
    \end{align}
\begin{example}
 For $M=1$, $-1\leq k,l,m, \leq 1$. In view of \eqref{eq:trip_vec} and \eqref{eq:kernels_vec}, the vectorized reduction in \eqref{eq:vector_DRP} requires the following $27\times 1$ vectors
     \begin{align}\label{eq:trip_vec_M1_example}
        \vect{t}_{n} = \begin{pmatrix}
         T_{-1, -1, -1}\\
         T_{-1,-1,0}\\
        \vdots\\
        T_{1,1,1}
        \end{pmatrix}, \hspace{0.1cm}
         \vect{s} = \begin{pmatrix}
         S_{-1,-1,-1}\\
         S_{-1,-1,0}\\
        \vdots\\
        S_{1,1,1}
        \end{pmatrix},
    \end{align}
and the triple sum in \eqref{eq:DRP_finite_x_y} is reduced to
    \begin{align}
        \sum_{k=-1}^1 \sum_{l=-1}^1\sum_{m=-1}^1 T_{klm} \cdot S_{klm} = \vect{t}_{{n}}^\mathrm{T} \vect{s}.
    \end{align}
\end{example}
 The expression \eqref{eq:vector_DRP} corresponds to the scenario when one polarization of a single input symbol $a_n$ and its triplets are used in \eqref{eq:single_ch_dual_pol_DRP_finite}. Considering that the model in \eqref{eq:single_ch_dual_pol_DRP_finite} is assumed to be stationary, we can take any instant of time as a reference to filling the input vectors of the optimizer. Without loss of generality, we take $n=0$. 
 
 To generate FRP outputs for a batch of samples using a single operation, we introduce a vectorization of the input-output pairs within the batch 
    \begin{align}\label{eq:in_vec}
        \vect{a} \triangleq (a_0, \cdots, a_B)^\mathrm{T},
    \end{align}
    \begin{align}\label{eq:outputs_vec}
        \vect{r} \triangleq (r_0,\cdots, r_B)^\mathrm{T},
    \end{align}
and a corresponding $B \times L$ matrix of triplets
\begin{align}\label{eq:triplets_library}
    \mathbb{T} \triangleq \Bigl(  \vect{t}_1^\mathrm{T}, \cdots, \vect{t}_B^\mathrm{T} \Bigr)^\mathrm{T}.
\end{align}
 In the data collection stage shown in  Fig.~\ref{fig:NBGD_sketch} (i), we summarize the notation defined in \eqref{eq:in_vec}, \eqref{eq:outputs_vec}, and \eqref{eq:triplets_library}. 
 
 Following the definitions in \eqref{eq:in_vec}--\eqref{eq:triplets_library}, the single polarization FRP prediction for a batch can be written as
 \begin{equation}\label{eq:channel_model_vector}
    \vect{r} \approx \vect{a} +\jmath\frac{8}{9} \gamma E_s \mathbb{T} \vect{s}, 
 \end{equation}
 where $\vect{s}$ is the vectors of kernels in \eqref{eq:kernels_vec}.
 
 \begin{example}
 For $M=1$ and $B=3$, three transmission pairs ${(a_0,r_0),(a_1,r_1),(a_2,r_2)}$ and their triplets $\vect{t}_{0}$, $\vect{t}_{1}$, $\vect{t}_{2}$, are considered. The matrix $\mathbb{T}$ in \eqref{eq:triplets_library} is therefore,  
 \begin{equation*}\label{eq:trip_vec_M1_example_b}
    \mathbb{T} = \Bigl(\vect{t}_0^\mathrm{T}, \vect{t}^\mathrm{T}_{1}, \vect{t}_2^\mathrm{T} \Bigr)^\mathrm{T} = 
    \setstacktabbedgap{1.5pt}
    \parenMatrixstack{
     t_{0,1}&t_{0,1}&\cdots&t_{0,27} \cr
     t_{1,1}&t_{1,2}&\cdots&t_{1,27} \cr
     t_{2,1}&t_{2,1}&\cdots&t_{2,27}},
\end{equation*}
where $t_{n,l}$ is the $l$-th triplet corresponding to $a_n$. Then, \eqref{eq:channel_model_vector} yields for this example 
\begin{align*}
    \begin{pmatrix}
     r_0\\
     r_1\\
     r_2
    \end{pmatrix} \approx 
    \begin{pmatrix}
     a_0\\
     a_1\\
     a_2
    \end{pmatrix} +
    \jmath\frac{8}{9} \gamma E_s    
   \setstacktabbedgap{1.5pt}
    \parenMatrixstack{
     t_{0,1}&t_{0,1}&\cdots&t_{0,27} \cr
     t_{1,1}&t_{1,2}&\cdots&t_{1,27} \cr
     t_{2,1}&t_{2,1}&\cdots&t_{2,27}}
    \begin{pmatrix}
     S_{-1,-1,-1}\\
     \vdots\\
     S_{1,1,1}
    \end{pmatrix}.
\end{align*}
\end{example}
 
 Between stages (i) and (ii), Fig.~\ref{fig:NBGD_sketch} shows explicitly that $\{\vect{a},\mathbb{T}, \vect{r}\}$ are the outputs of the data collection block (i) that are taken by the optimizer update stage. During the optimizer update, there are two sub-stages. First, an NBGD step towards minimizing the objective function is taken. Second, the FRP prediction is calculated using the NBGD kernels output of the optimization step. The latter is obtained through \eqref{eq:channel_model_vector} using
 \begin{equation}\label{eq:FRP_iteration}
    \hat{\vect{r}}^{(l+1)} = \vect{a} + \jmath\frac{8}{9} \gamma E_s \mathbb{T} \hat{\vect{s}}^{(l+1)}.  
 \end{equation}
 The hats in \eqref{eq:FRP_iteration} (and more generally throughout this paper) denote approximated values using the NBGD algorithm. 
 
 We consider as an objective function the mean squared error (MSE), defined as 
 \begin{equation}\label{eq:objective}
    \text{MSE}(\hat{\vect{r}},\vect{r}) \triangleq \frac{1}{2B} ||\hat{\vect{r}}-\vect{r} ||^2.
 \end{equation}
 In \eqref{eq:objective}, $\vect{r}$ denotes the true output batch, while $\hat {\vect{r}}$ corresponds to the FRP prediction using the NBGD algorithm. A factor of $\frac{1}{2}$ is added in the MSE definition in \eqref{eq:objective} to simplify its gradient expression, see App.~\ref{app:wirtinger}. We stress that the definition of $\text{MSE}(\cdot)$ in \eqref{eq:objective} carries an implicit dependency on the optimized SPM kernels through $\hat{\vect{r}}$ in \eqref{eq:FRP_iteration}. 

 Now that we have introduced the variables, the batch vectorized notation, the objective function, and the overall schedule of the proposal optimizer, we provide the analytical expression to compute the optimization step for the NBGD (see \eqref{eq:norm_gd}). This expression is given in the following theorem.    
 
 \begin{theorem}\label{T:NBGD}
 The NBGD descent step for the objective function in \eqref{eq:objective} is\footnote{Note that $|\cdot|$ in the r.h.s. of \eqref{eq:NBGD_main} denotes again an element-wise absolute value.}
 \begin{align}\label{eq:NBGD_main}
   \hat{\vect{s}}^{(l+1)} = \hat{\vect{s}}^{(l)}+ \jmath \alpha \mathbb{T}^{\dagger}(\hat{\vect{r}}^{(l)}-\vect{r})\oslash|\mathbb{T}^{\dagger}(\hat{\vect{r}}^{(l)}-\vect{r})|,
 \end{align}
 where $\hat{\vect{s}}^{(l)}$ is the kernels' vector at iteration $l$, $\alpha$ is the step size, $\mathbb{T}$ is the $B \times L$ matrix of triplets in \eqref{eq:triplets_library}, $\vect{r}$ is the vector of true outputs, and $\hat{\vect{r}}^{(l)}$ is the FRP prediction of $\vect{r}$ using the NBGD algorithm at iteration $l$ in \eqref{eq:FRP_iteration}.
 \end{theorem}
 \begin{proof}
    See Appendix~\ref{app:wirtinger}.
 \end{proof}
  As shown in Fig.~\ref{fig:NBGD_sketch} (ii), the NBGD block yields a vector of optimized kernels $\hat{\vect{s}}$ that together with $\{\vect{a},\mathbb{T}\}$ are sent to the FRP block. Once $\hat{\vect{r}}^{(l)}$ is generated, the final stage, shown in Fig.~\ref{fig:NBGD_sketch} (iii), is the convergence assessment. During this stage, the performance of the model is checked against the SSFM simulation. True and approximated outputs $(\vect{r},\hat{\vect{r}}^{(l)})$ are used to compute the objective function \eqref{eq:objective}. If the MSE is larger than a predefined threshold $\tau$, an increment in the NBGD iteration is sent to the optimizer update block (ii). Otherwise, the optimization continues until the MSE is lower than $\tau$ or a predefined number of iterations is reached. Although the proposed optimization scheme uses only one polarization, the FRP expression for both is equivalent since they share the same linear dependency with kernels. Therefore, the optimized kernels can be used to obtain a similar performance for the second polarization.   

 Linear regression as a strategy to optimize FRP has been previously addressed in works targeting perturbation-based nonlinearity compensation techniques, for example, \cite{Tao2011, Liang2014,Rafique2015}. Nevertheless, this paper is introducing a linear optimization approach to enhance FRP from a modeling perspective. Although multiple optimization techniques can be adopted and the ordinary least squares solution (OLS) is the optimal solution when MSE is the objective function, there are two reasons why gradient descent was chosen. First, OLS exhibits a well-known tendency to overfit the optimization batch, which results in a failure to accurately fit new batches. Secondly, for large values of $M$ the matrix of triplets may yield an underdetermined system of linear equations when low-cardinality constellations are used. To overcome the potential above-mentioned issues of OLS, we chose GD as an optimization strategy.
\subsection{NBGD Implementation Aspects} \label{sec:implementation}
The NBGD algorithm was implemented with the set of parameters specified in Table \ref{table:nbgdparm}.~In this subsection, we summarize key aspects of the numerical implementation. 

\noindent \textit{Data collection stage:} For a fixed power $P$ and memory $M$, a stream of transmission pairs is generated according to Fig.~\ref{fig:system_model}. Using a sliding window, the batch vectors $\vect{a}$ and $\vect{r}$ are populated, and the corresponding input neighbors are processed to fill in the matrix of triplets $\mathbb{T}$. The batch size $B$ is left as a tuning parameter to adjust the computational demand of the optimizer.

\noindent \textit{Optimizer update:} The estimated kernels vector is initialized to zero\footnote{No significant advantage was observed when initializing with a non-zero vector. There was no good reason to prefer one way of initializing over the other, thus we simply chose to initialize all kernels to zero.}, while for choosing $\alpha^{(0)}$ two criteria are considered: order of magnitude and magnitude itself. Firstly, the $\alpha$'s order of magnitude should match the order of the gradient. Choosing the right order is crucial since NBGD can meander around or slowly crawl near stationary points. Secondly, a \emph{scheduled learning rate} is chosen. The scheduled decreases $\alpha^{(l)}$ a fraction of it after $15$ iterations. This scheduled learning rate is a homogeneous staircase down when plotted with respect to $l$. 

Once the gradient step is performed, the batch of inputs and triplets $\{\vect{a},\mathbb{T}\}$ is sent to the FRP block together with current NBGD kernels' vector $\hat{\vect{r}}^{(l)}$ to generate $\hat{\vect{r}}^{(l+1)}$ (see Fig.~\ref{fig:NBGD_sketch} stage (ii)). Then the estimated output vector is delivered to stage (iii) in Fig.~\ref{fig:NBGD_sketch}.

\noindent \textit{Convergence assessment:} MSE is computed using the SSFM outputs and estimated outputs of the batch $(\vect{r}, \hat{\vect{r}})$ and a threshold check is performed. The latter is such that when the difference between the current and previous MSE values is higher than the threshold and $l$ is smaller than a user-defined max number of iterations $\bar{l}$, a new gradient step is taken and the cycle is repeated. If not, the algorithm converges and the reached MSE value and the NBGD kernels are stored as MSE$^c$ and $\vect{s}^c$, respectively, and the algorithm moves to the validation stage.

\noindent \textit{Convergence Validation:} Once NBGD has converged, a validation test is performed over a fresh batch $\{\vect{a},\mathbb{T},\vect{r}\}^\star$. If MSE for the validation batch is $\leq 10\text{MSE}^c$, we consider NBGD to have converged, if not,  $\alpha^{(0)}$ is reduced $25\%$ in magnitude, and NBGD is restarted with $\hat{\vect{s}}^{(0)}=\vect{s}^c$. The validation stage is done to avoid overfitting and it is performed a maximum of 5 times. If within these 5 iterations NBGD doesn't converge, the algorithm ends with no estimation reached. 

\begin{table}[t!]
\def\arraystretch{0.9}
\caption{NBGD algorithm parameters}
\centering
{\footnotesize
\begin{tabular}{c c}
\toprule 
Kernel initialization        &  $\hat{\vect{s}}^{(0)}=0+0\jmath$ \\
Gradient step initialization &  $\alpha^{(0)} \approx 1*\mathcal{O}(\nabla \text{MSE})$ \\
Schedule function            &  $\alpha^{(l)} = 0.9\alpha^{(l-1)}$ iif $\lfloor l/15 \rfloor \in \mathbb{Z}^+$ \\
Threshold                    &  $0.1 \alpha$ \\
Max number of iterations $\bar{l}$ & $10^5$  \\
\bottomrule
\end{tabular}}
\label{table:nbgdparm} 
\end{table}

\section{NBGD Numerical Results}\label{sec:nbgd_results}
 In order to validate the NBGD optimized model of the previous section, a set of simulations in terms of the metrics discussed in Sec. \ref{sec:frp_performance} was performed. This section considers the same representative transmission scheme introduced in Sec. \ref{sec:frp_performance}. Nevertheless, the NBGD algorithm can be used to optimize kernels for multiple transmission scenarios. Additionally, the end of this section includes a simple complexity analysis. Throughout this section, we refer to the FRP model computed with NBGD-optimized kernels simply as NBGD.
 
\subsection{Numerical Validation}
 In analogy to Fig.~\ref{fig:snr_vs_p}, Fig.~\ref{fig:snr_vs_p_nbgd} shows SNR versus input power. This figure shows the NBGD results without ASE for two memory sizes ($M =0,9$), with ASE for $M=9$, as well as the numerical simulations (SSFM) with and without ASE.  Overall, in the results shown in Fig.~\ref{fig:snr_vs_p_nbgd} it is evident that NBGD is in good proximity in magnitude and overall behavior to SSFM for all powers in the region of study with and without ASE. Unlike the results in Fig.~\ref{fig:snr_vs_p}, NBGD exhibits a small gap to the SSFM baseline, even for the extreme case $M=0$. Previously, we showed in Fig.~\ref{fig:snr_vs_p} that FRP starts to diverge from SSFM at $10$~dBm. For that power and $M=15$, the gap between FPR and SSFM is approximately $0.39$~dB. As shown in Fig.~\ref{fig:snr_vs_p_nbgd}, NBGD reaches the same proximity to the SSFM value only at $P=$17~dBm, for $M=9$, thus yielding a $7$~dB extension with respect to FRP of the model's range of validity. In Fig.~\ref{fig:snr_vs_p_nbgd} we highlight that NBGD's accuracy extends up to $10$~dB above the optimum launch power, surpassing the $3$~dB achieved by FRP shown in Fig.~\ref{fig:snr_vs_p}.
 
\begin{figure}[!t]
    \centering 
    \resizebox{0.5\textwidth}{!}{\includegraphics{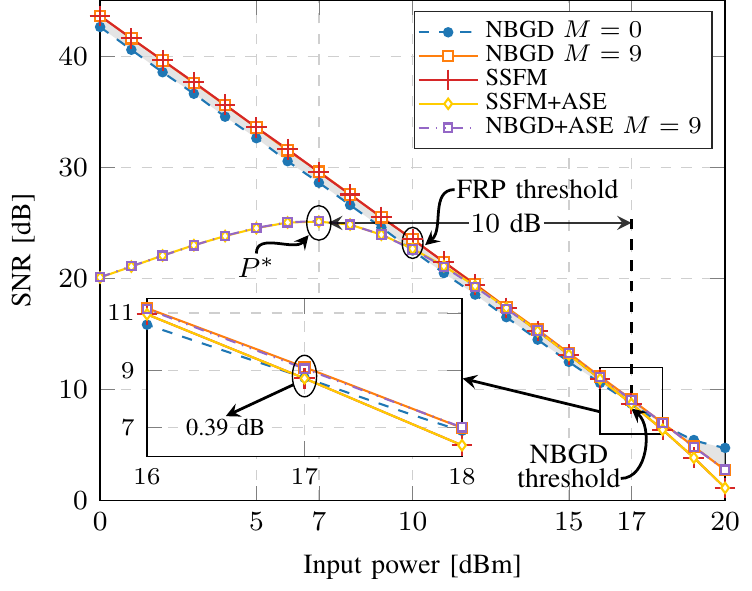}}
    \caption{The SNR in \eqref{eq:snr_def} as a function of launch power $P$ with (``+ASE'') and without ASE noise. Increments in memory size close the gap between SSFM and NBGD in absence of ASE noise. The optimum launch power for the ``+ASE'' case $P^*=7$~dBm is shown.  The gray area delimits the region covered for $0\leq M \leq 9$. The NBGD prediction (with and without ASE) starts to fail at about $10$~dB above $P^*$.} \label{fig:snr_vs_p_nbgd}
\end{figure}

\begin{figure*}[!t]
\centering
\resizebox{0.32\textwidth}{!}{\includegraphics{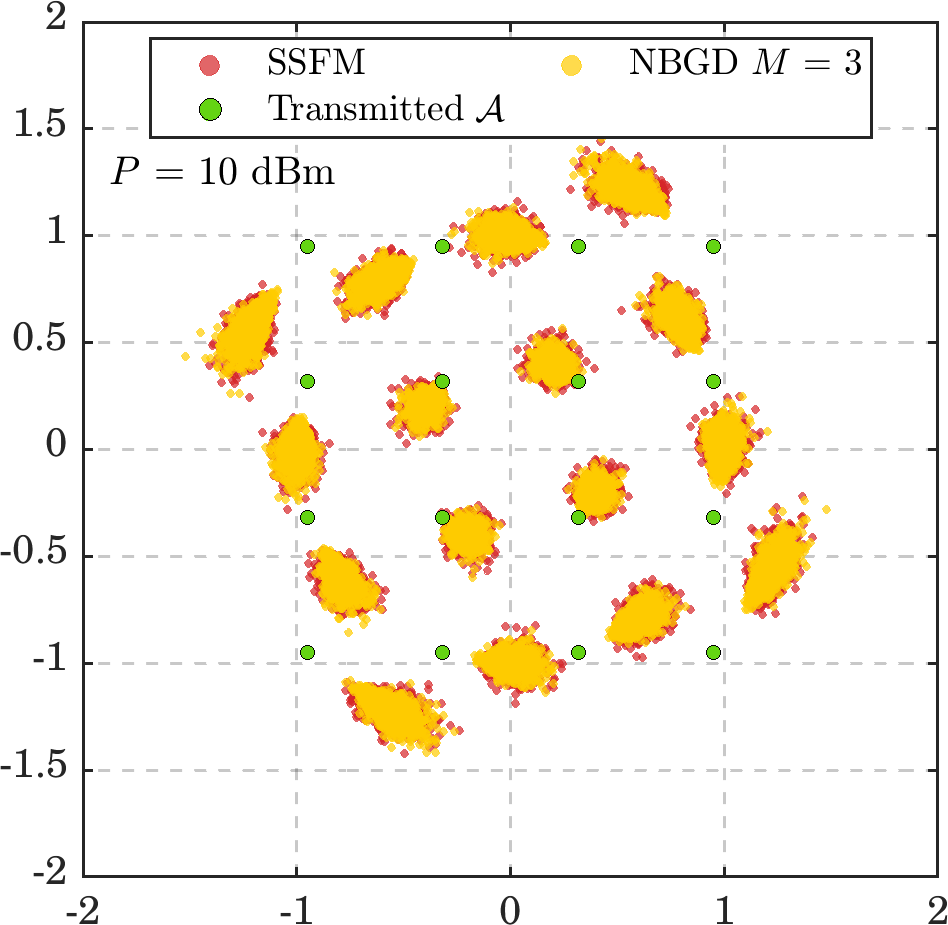}}
\hfill
\resizebox{0.32\textwidth}{!}{\includegraphics{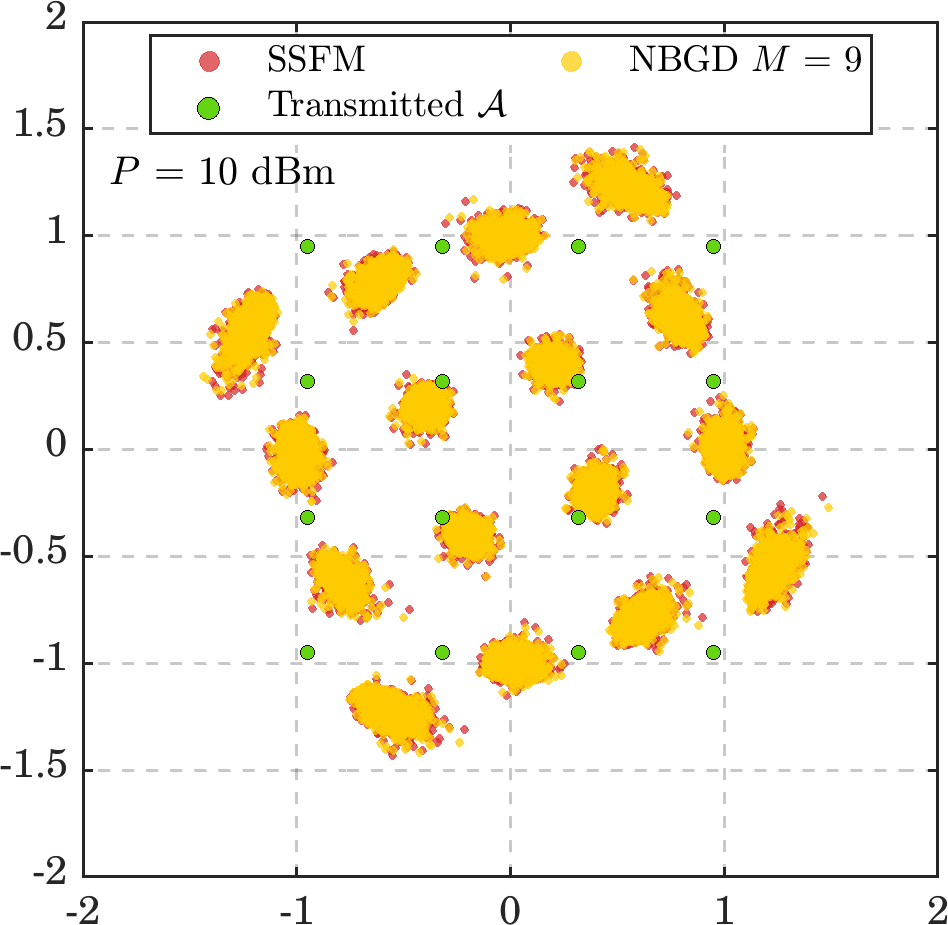}}
\hfill
\resizebox{0.32\textwidth}{!}{\includegraphics{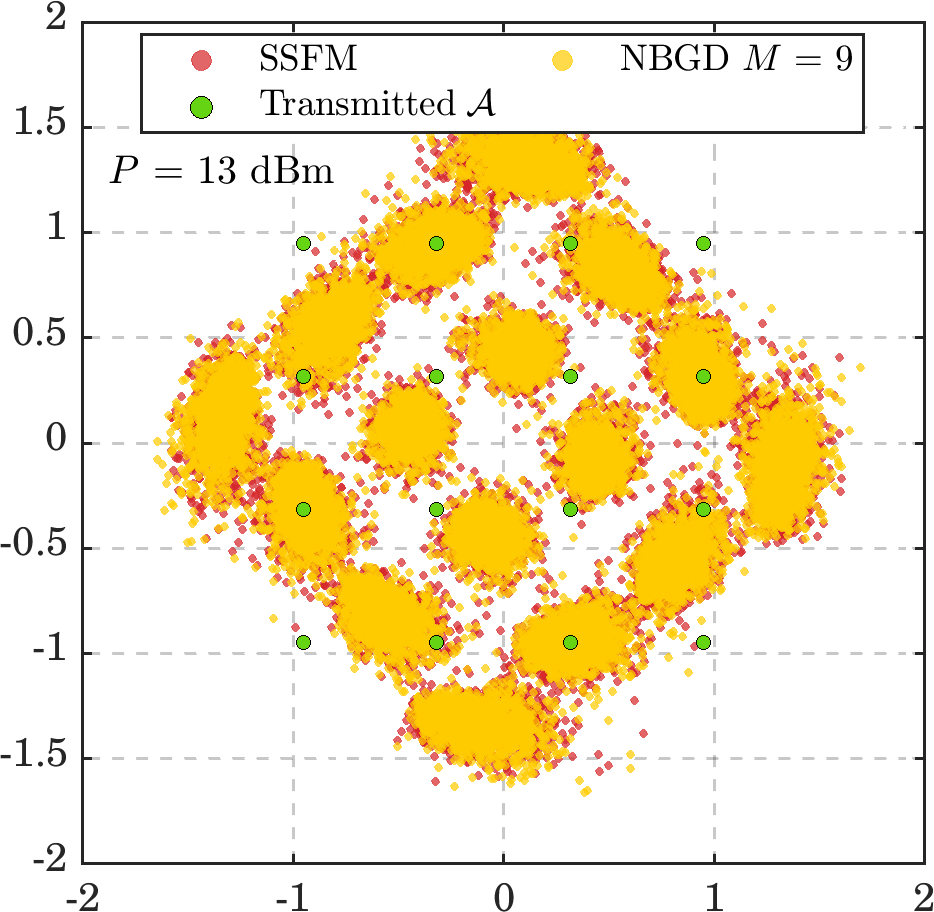}}
\caption{NBGD constellation diagrams for (a) $P=10$~dBm and NBGD $M=3$, (b) $P=10$~dBm and NBGD $M=3$, and (c) $P=13$~dBm and NBGD $M=9$. Overall, the mismatch between the constellations observed in Fig.~\ref{fig:clouds} is improved. For $P=$10~dBm, the match is satisfactory at $M=3$ (a) but enhanced with memory size increment to $M=9$ (b).} \label{fig:nbgd_clouds}
\end{figure*}
 
 Fig.~\ref{fig:nbgd_clouds} shows the constellation diagrams of the two representative powers considered in Fig. \ref{fig:clouds} ($10$ and $13$~dBm). These two powers are within the NBGD validity region. The constellation diagram at $P=10$~dBm and memory size $M=3$ in Fig.~\ref{fig:nbgd_clouds} (a) shows that NBGD matches SSFM better than FRP with $M=5$ in Fig.~\ref{fig:clouds} (a). In addition, the NBGD clouds for $M=9$ shown in Fig.~\ref{fig:nbgd_clouds} (b) and (c) exhibit a good match in phase and amplitude scaling. In Fig.~\ref{fig:mu_trajectory_nbgd}, we show the evolution of the NBGD conditional means $\hat{\mu}$ as a function of $M$. We consider the same powers chosen in Fig.~\ref{fig:mu_trajectory} and show the first quadrant of $16$~QAM. For both powers, it can be observed that unlike the results in Fig.~\ref{fig:mu_trajectory}, the conditional means for NBGD quickly converge to values falling on the same rings as SSFM. Overall, the $\hat{\mu}_m$ values are in better proximity to all $\mu_m$ than FRP is, as observed in Fig. \ref{fig:mu_trajectory}. 

 During the assessment of NBGD, it was observed that for powers below the FRP threshold, the optimized kernels at a power level within this regime yield good accuracy for any power within the FRP's validation region. This behavior is expected since, in the linear and pseudo-linear regimes, the cloud's shape does not significantly change with power increments.  For power above the FRP threshold, point-to-point optimization is required. Nevertheless, this does not represent a critical disadvantage for NBGD with respect to FRP, especially because NBGD reaches good accuracy with a small $M$.
 
\begin{figure}[!t]
    \centering 
    \resizebox{0.48\textwidth}{!}{\includegraphics{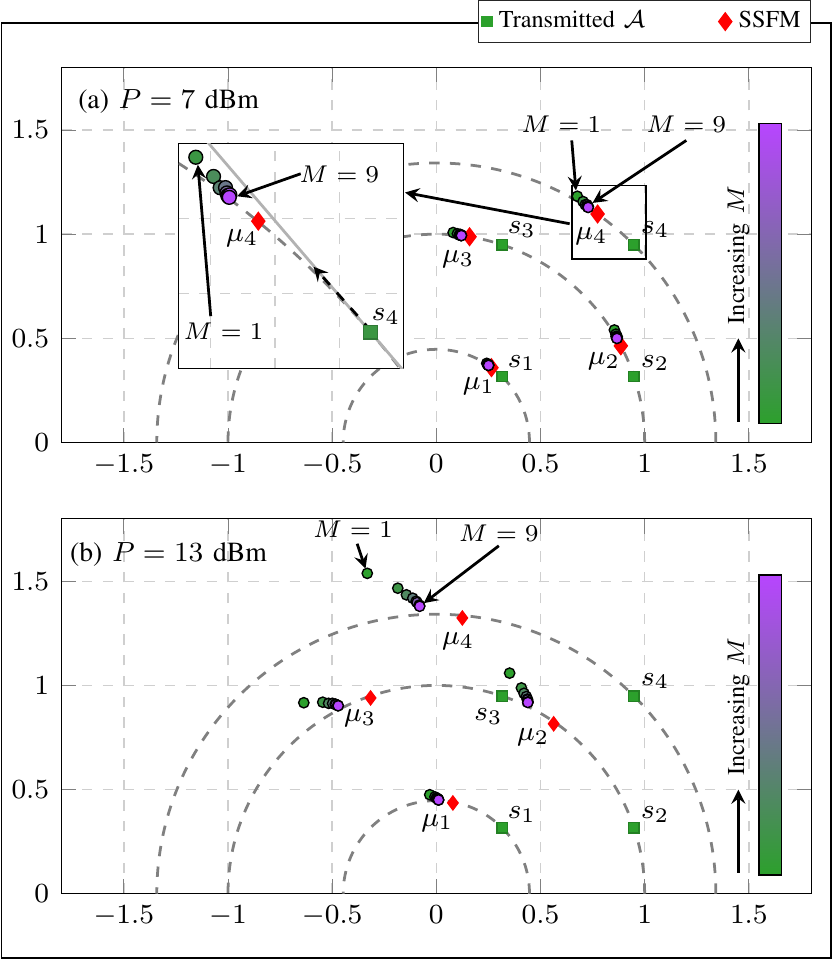}}
    \caption{The evolution with increasing memory $M$ for $\hat{\mu}$ in \eqref{eq:cond_mean_frp} is depicted with a color gradient. We consider the four symbols $\{s_1,s_2,s_3,s_4\}$ within the first quadrant of $16$-QAM, and two powers (a) $P=7$~dBm ($P^*$) and (b) $P=13$~dBm (example Fig.~\ref{fig:clouds} (c)). The rings where $\{s_m\}$ lay are drawn and ${\mu_m}$, red dimonds,  denote the SSFM conditional means for $s_m$. It can be observed that overall $\{\hat{\mu}_m\}$ converges to the $\{s_m\}$ rings, and they do not match $\{\mu_m\}$ for either (a) nor (b).} \label{fig:mu_trajectory_nbgd}
\end{figure}

 The performance of the NBGD model is also examined in terms of the two metrics $\Delta r$ and $\Delta \phi$ defined in \eqref{eq:radii_ratio} and \eqref{eq:phase_diff}, respectively. The results are shown in Fig.~\ref{fig:summary_radii_phase_nbgd} (a) and (b), respectively, for three representative memory sizes: $M=0,3,9$. On one hand, we observe in Fig.~\ref{fig:summary_radii_phase_nbgd} (a) that the optimized model resulting from the NBGD (i) shows a good fit to SSFM in the linear and pseudo-linear regimes, and (ii) exhibits an opposite behavior to the trend shown in Fig.~\ref{fig:summary_radii_phase} (a): instead of diverging from SSFM with memory size increments, the optimized model moves towards the SSFM simulations. This behavior is the result of a better prediction of the conditional means, and thus, of the received power. In Fig.~\ref{fig:summary_radii_phase_nbgd} FRP $M=0$ is shown as the FRP reference since is the closest scenario to SSFM. On the other hand, we observe in Fig.~\ref{fig:summary_radii_phase_nbgd} (b) that the optimized model outperforms FRP with $M=15$ even for the extreme case $M=0$. Therefore, NBGD enhances the phase match. For any $M\geq3$ the model matches the SSFM baseline, and thus, we can generate an accurate prediction of the average nonlinear phase rotation with a very low model memory size. In addition, Fig.~\ref{fig:summary_radii_phase_nbgd} shows the NBGD threshold found in Fig.~\ref{fig:snr_vs_p_nbgd}. In the threshold, both, $\Delta r$ and $\Delta \phi$ show good proximity for $M\geq3$. In the particular case of these two metrics, we observe a $7$~dB gain with respect to the FRP reference.
\begin{figure*}[!t]
    \centering 
    \centering
    \resizebox{0.49\textwidth}{!}{\includegraphics{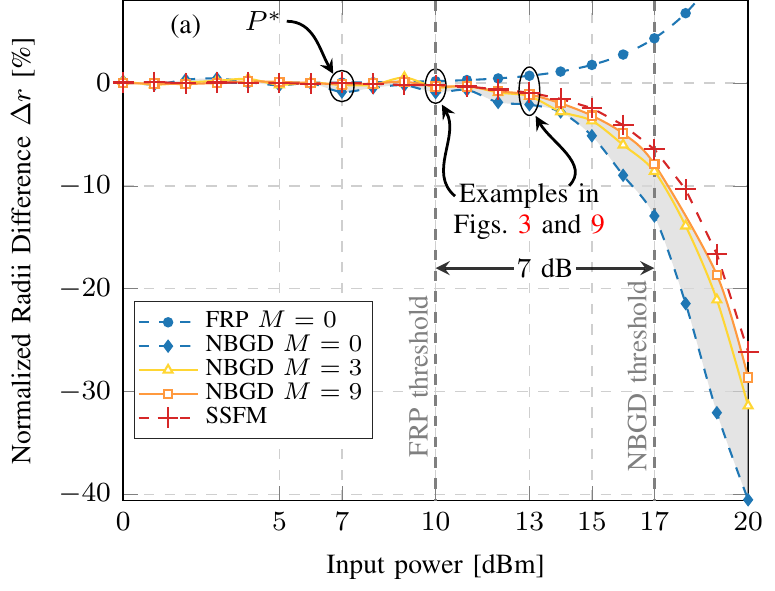}}
    \hfill
    \resizebox{0.49\textwidth}{!}{\includegraphics{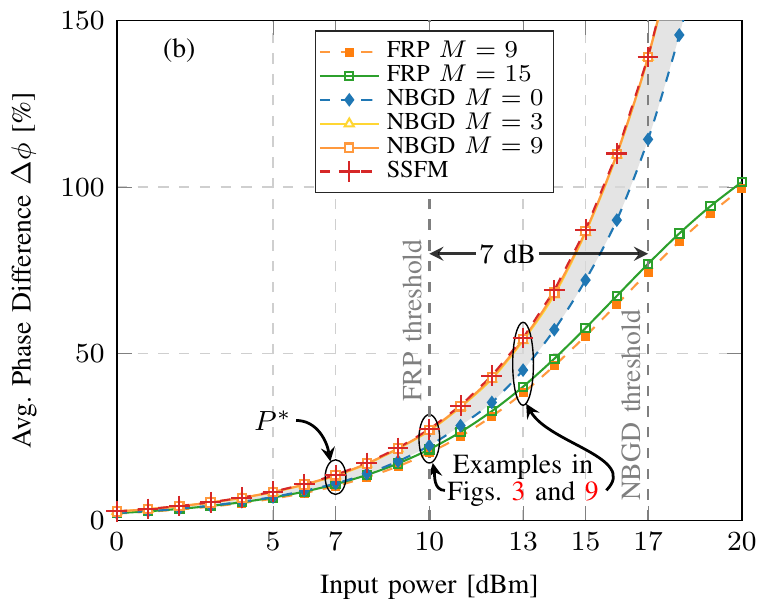}}
    \caption{(a) Normalized radii difference, (b) phase difference, and. The gray areas delimit the region covered for $0\leq M \leq 9$, and $M = 3$ is shown explicitly. The plots include the simulations with SSFM, as well as the FRP scenario that is the closest to SSFM.} \label{fig:summary_radii_phase_nbgd}
\end{figure*}
 
 Fig~\ref{fig:relative_error_nbgd} (a) shows the relative error defined in \eqref{eq:relateverror} as a function of power. The results for NBGD show once more an overall improvement with respect to the best FRP baseline ($M=15$ in this case). For memories $M>0$, the NBGD kernels allow the FRP model to significantly decrease the relative error in the region of powers beyond the conventional threshold ($11$\%). All the NBGD curves are below the FRP and they exhibit a decreasing trend with increments in memory size. Fig.~\ref{fig:relative_error_nbgd} (a) shows that NBGD with $M=9$ reaches an 11\% relative error at $16$~dBm. This represents a 6~dB gain compared to FRP using the same memory size. This result, also supported by the SNR analysis done in Fig.~\ref{fig:snr_vs_p_nbgd}, confirms that NBGD significantly extends the region of validity of the FRP model.
     
\begin{figure*}[!t]
    \centering 
    \resizebox{0.49\textwidth}{!}{\includegraphics{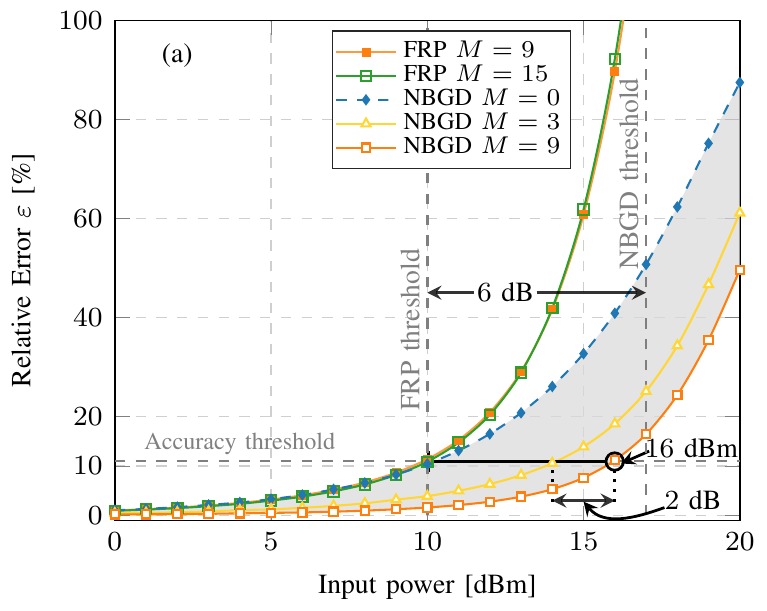}}
    \hfill
    \resizebox{0.47\textwidth}{!}{\includegraphics{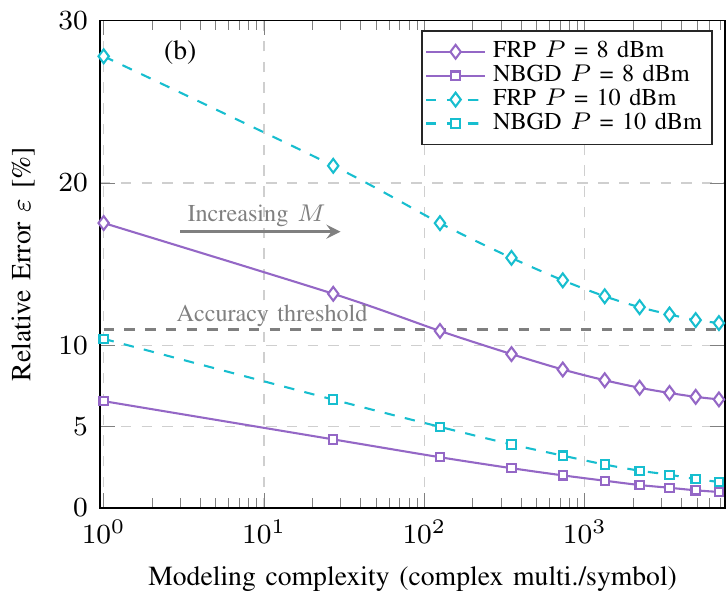}}
    \caption{Relative error of FRP/NBGD with respect to SSFM as a function of (a) power and (b) complex multiplications per symbols. The gray area in (a) is delimiting the region covered for $0\leq M \leq 9$, and $M = 3$ is shown explicitly. In (b) $P=8$~dBm and $10$~dBm are considered. Notice in (b) that increments in memory size monotonically decrease the error and that FRP reaches the accuracy threshold at higher complexity than NBGD does}. \label{fig:relative_error_nbgd}
\end{figure*}

Finally, given that a multi-span setup does not change the FRP's mathematical formalism in (3), we can argue that the NBGD approach paradigm remains applicable to the multi-span scenario and that the optimization process will be able to capture the nonlinear effects as effectively as it does for the single-span scenario. A study of the performance of NBGD for multi-span systems is out of the scope of this work and is left for further investigation.

 \subsection{Modelling Complexity Reduction}
 Overall, the results shown in Figs~\ref{fig:snr_vs_p_nbgd}, \ref{fig:mu_trajectory_nbgd}, and \ref{fig:summary_radii_phase_nbgd} prove a significant reduction in the memory size needed to guarantee a satisfying degree of accuracy using NBGD. In Fig.~\ref{fig:relative_error_nbgd} (a), it can also be observed that with $M=3$ the accuracy threshold is now achieved for $P=14$~dBm and not for $8$~dBm as shown in Sec.~\ref{subsec:RE}. Using the same reasoning used in Sec.~\ref{subsec:RE}, we observe that also for NBGD is possible to prolong the accuracy reached with $M=3$. In the NBGD case, increasing $M$ from $3$ (343 kernels) to $9$ (6 859 kernels) prolongs $2$~dB the model's accuracy. The NBGD optimization is therefore providing larger increments of the model's accuracy, and it does it at a cheap computational cost. Using $M=9$ instead of $M=15$ represents a reduction of 78\% in the number of kernels needed to hold $\varepsilon$ in the accuracy threshold.  

 To further analyze the memory reduction, we provide in Fig.~\ref{fig:relative_error_nbgd} (b) an evaluation of the relative error as a function of the modeling complexity, i.e., the number of kernels needed to compute the FRP model. We consider the powers $P=8$~dBm and $P=10$~dBm and evaluate $\varepsilon$ for FRP and the NBGD. The chosen powers are within the pseudo-linear regime. Fig.~\ref{fig:relative_error_nbgd} (b) shows, as expected, that FRP and NBGD have a decreasing trend with increments in complexity. The reached $\varepsilon$ values are power-dependent and they worsen when the system approaches the nonlinear regime. It is evident from Fig.~\ref{fig:relative_error_nbgd} (b), that NBGD outperforms FRP. NBGD achieves relative errors below the accuracy threshold at a very low complexity cost. NBGD is therefore providing a satisfactory level of accuracy using small memory sizes as hinted by Fig.~\ref{fig:snr_vs_p_nbgd} and Fig.~\ref{fig:summary_radii_phase_nbgd}. We can conclude that NBGD requires fewer kernels to compute the model (see \eqref{eq:numk}), which translates into a significant reduction in the computational complexity needed to generate accurate predictions. 

 Another relevant aspect to comment on is the computational complexity required for the FRP and NBGD kernels reported in this work. On one hand, to compute the kernels in their integral form we performed a simple Riemann integration. The number of multiplications and sums demanded by this approximation was $0.68\cdot10^6\cdot(2M+1)^3$, for which we considered the minimum number of steps leading to sufficient accuracy. On the other hand, the NBGD optimizer demanded $3\cdot N^*_{\text{it}}\cdot B\cdot(2M+1)^3$ multiplications and sums, where $B$ is the NBGD batch size, and $N^*_{\text{it}}$ is the average number of iterations for NBGD convergence ($\approx 2\cdot10^2$ ). Tuning the batch size $B$ has a direct impact on the complexity. In general, the chosen batch size does not depend on $M$, thus considering $B$ large enough but smaller than $\mathcal{O}(M^3)$, allows the NBGD's complexity to remain significantly below the FRP's even when using the same value of $M$. However, we reiterate that the advantage of NBGD is that it reduces the model's memory $M$ required to achieve a fixed level of accuracy with respect to FRP, hence further increasing the complexity gap with FRP.

 In Fig.~\ref{fig:generalization_nbgd}, we present an evaluation of the relative error as a function of the modeling complexity, for two additional modulation formats, QPSK and 64QAM. The relative error for these cases is computed using the NBGD kernels optimized for 16QAM. The results shown in Fig.~\ref{fig:relative_error_nbgd} (b) for FRP and NBGD, are displayed in Fig.~\ref{fig:generalization_nbgd} (a) and (b) as a baseline for the analysis of the modulation format generalization of NBGD kernels. As observed for the 16QAM case, QPSK and 64QAM outperform FRP for both considered powers. Their overall behavior across the complexity range resembles 16QAM as well, and there are no significant penalties for the relative error at low complexity. These results allow us to conclude that the NBGD kernels are generalizing well to other QAM modulation formats. 

\begin{figure*}[!t]
    \centering 
    \resizebox{0.49\textwidth}{!}{\includegraphics{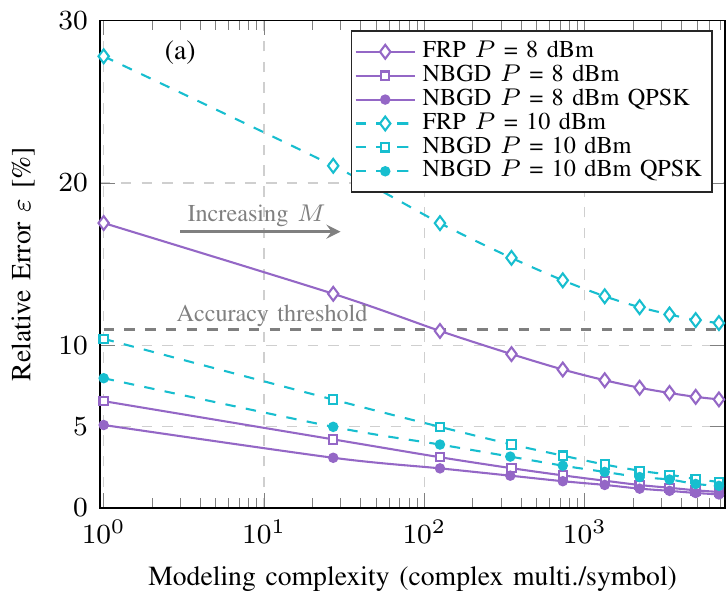}}
    \hfill
    \resizebox{0.49\textwidth}{!}{\includegraphics{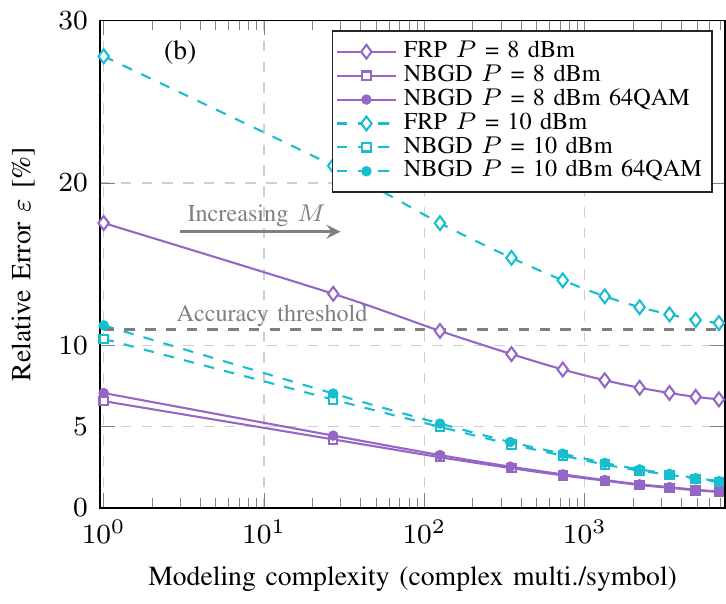}}
    \caption{Using the NBGD kernels optimized with a 16QAM modulation format, the relative error of NBGD is plotted with respect to SSFM as a function of complex multiplications per symbols for (a) QPSK and (b) 64QAM. The results of Fig.~\ref{fig:relative_error_nbgd} (b) are included in (a) and (b) to provide a benchmark for comparison. Overall, the results for QPSK and 64QAM confirm that the 16QAM NBGD kernels generalize well to other QAM formats.} \label{fig:generalization_nbgd}
\end{figure*}

 As a final remark, notice that formally, both, FRP and NBGD are infinite-memory models since the nonlinear kernels only vanish for $M\to\infty$. However, an effective finite value of $M$ for which the perturbation model can be considered accurate enough can be defined. One way to do this is via the introduction of ad-hoc accuracy thresholds, as we did in Sec. III. The key outcome of our paper is that an effective memory reduction is introduced by NBGD, which holds regardless of the specific threshold set.
\section{Conclusions}
We conducted a comprehensive investigation on the accuracy and limitations of the time-domain FRP model for a single-span, single-channel transmission in dual-polarization systems. We proposed a data-driven enhancement of FRP (called NBGD) by optimizing the perturbation kernels. Various numerical simulations were conducted to demonstrate that NBGD provides the same accuracy as FRP with a reduced memory for the model, therefore, reducing the computational complexity needed to generate a satisfactory prediction. In addition, we have shown that NBGD extends the power range of validity of the FRP model above the pseudo-linear threshold. For the study case considered, the extension is $6-7$~dB above the chosen accuracy threshold for the FRP, depending on the metric used. We also show that the NBGD's good performance generalizes well for QAM modulation formats.

Future works include the extension of the NBGD method to a multi-span wavelength-division multiplexed scenario, and developing nonlinearity compensation/mitigation algorithms based on the enhanced low-complexity model resulting from NBGD. Experimental validation is also left for future work.

\section*{Acknowledgments}
The authors would like to thank Prof. Erik Agrell (Chalmers University of Technology) and  Dr. Olga Vassilieva (Fujitsu Network Communications, Inc.) for comments on earlier versions of this manuscript. The authors would also like to acknowledge Dr.-Ing. Tobias Fehenberger (ADVA Optical Networking SE) for providing an early version of the numerical computation of kernels in their integral form. 
\appendices
\section{Proof of Theorem \ref{T:conditional_mean}}\label{app:condmean}
Let $A_{\xpol/\ypol}$ and $R_{\xpol/\ypol}$ be the complex random variables (RV) associated with one of the $\xpol/\ypol$ polarization components of the transmitted and received symbols, respectively. 
Under the assumption that \eqref{eq:single_ch_dual_pol_DRP_finite} is a stationary channel, we take the conditional expectation of one of its components ($\xpol$) for $n=0$
\begin{align}\label{eq:ex_rp_rvs}
    \tilde\mu_{\xpol,m}  = s_m + \mathbb{E}\{\Delta A_{\xpol,0}|A_{\xpol,0}=s_m\}.
\end{align}
 To determine $\mathbb{E}\{\Delta A_{\xpol,0}|A_{\xpol,0}=s_m\}$ we consider four possible cases of the indices $(k,l,m)$ in \eqref{eq:single_ch_dual_pol_DRP_finite2}:
\begin{enumerate}
    \item \textbf{Case I:} \textit{all indices are  different and nonzero}, i.e., $k\neq l\neq m \neq 0$. 
    \item  \textbf{Case II:} \textit{a single index is zero}, i.e, $k = 0$ and $l\neq0, m \neq 0$; or $l = 0$ and $k\neq 0, m \neq 0$; or $m = 0$ and $k\neq 0, l\neq 0$. Here two sub-cases arise: (i) when the remaining two indices are identical; and (ii) the remaining two indices are different. 
    \item  \textbf{Case III:} \textit{two indices are equal to zero and the third one is nonzero}, i.e., $k=l=0,m\neq 0$, or $k=m=0,l\neq 0$, or $l=m=0, k\neq 0$. 
    \item  \textbf{Case IV:} \textit{all indices are zero}, i.e. $k = l = m = 0$. 
\end{enumerate}

In the following, we compute each of the contributions given by the above cases.

\subsection{Case I}
In this case, the contribution is given by
\begin{align}\label{App_eq:caseI}
    \mathbb{E}&\{\Delta A_{\xpol,0}|A_{\xpol,0}=s_m\} \nonumber\\ & = \mathbb{E}\Bigg\{ c\sum_{(k,l,m)\in\mathcal{S}\setminus \{0\}} (A^*_{\xpol,k}A_{\xpol,l}+A^*_{\ypol,k}A_{\ypol,l})A_{\xpol,m}S_{klm}\Bigg\},
\end{align}
where $c\triangleq \jmath \frac{8}{9}\gamma E_s$ and $\mathcal{S}$ is defined in \eqref{eq:DRP_2}. By virtue of the independence of the random variables $A_{\xpol,n}$ and $A_{\ypol,n}$ $\forall n$, the expectation distributes over all sums and products in \eqref{App_eq:caseI}. In addition, due to the zero-mean condition assumed for $\mathcal{A}$ in Sec.~\ref{sec:snr_def} 
this scenario leads to zero contribution. 

\subsection{Case II}
The contributions are in this case
\begin{align}\label{App_eq:caseI_ia}
    \mathbb{E}&\{\Delta A_{\xpol,0}|A_{\xpol,0}=s_m\} = \nonumber\\ &  \mathbb{E}\Bigg\{ c\sum_{l\in\mathcal{W}}\sum_{m\in\mathcal{W}} (s_m^*A_{\xpol,l}+A^*_{\ypol,0}A_{\ypol,l})A_{\xpol,m}S_{0lm}\Bigg\} \nonumber \\
     & + \hspace{.1cm} \mathbb{E}\Bigg\{ c\sum_{k\in\mathcal{W}}\sum_{m\in\mathcal{W}} (A^*_{\xpol,k}s_m+A^*_{\ypol,k}A_{\ypol,0})A_{\xpol,m}S_{k0m}\Bigg\} \nonumber \\
     & + \hspace{.2cm} \mathbb{E}\Bigg\{ c\sum_{k\in\mathcal{W}}\sum_{l\in\mathcal{W}} (A^*_{\xpol,k}A_{\xpol,l}+A^*_{\ypol,k}A_{\ypol,l})s_m S_{kl0}\Bigg\},
\end{align}
Where $\mathcal{W} = \{i\in[-M,M]\setminus \{0\}\}$. 

For the first term in the r.h.s. of \eqref{App_eq:caseI_ia} subcase(i), i.e., when $l\neq m\neq 0$, the contribution is zero due to the RVs independence, similarly to Case I. For sub-case (ii), i.e., when $l=m\neq0$, the first term in \eqref{App_eq:caseI_ia} is reduced to
\begin{align}\label{App_eq:caseI_ia_b}
    c s_m^* \sum_{l\in\mathcal{W}} \mathbb{E}\{A^2_{\xpol,l}\}S_{0ll} = c s_m^* \mathbb{E}\{A^2_{\xpol}\} \sum_{l\in\mathcal{W}}S_{0ll},
\end{align}
where we have not included all contributions that are zero due to the RVs' zero-mean assumption.

For the second term in \eqref{App_eq:caseI_ia} sub-case (i), i.e., for $k\neq m\neq 0$, the contribution is zero by virtue of the RVs independence. For the sub-case (ii), when $k=m\neq0$, the second term in \eqref{App_eq:caseI_ia} is
\begin{align}\label{App_eq:caseI_ib_b}
    c s_m\mathbb{E}\Bigg\{ \sum_{k\in\mathcal{W}} |A_{\xpol,k}|^2S_{k0k}\Bigg\} =  c s_m \sum_{k\in\mathcal{W}}S_{k0k},
\end{align}
where all contributions that are zero by virtue of the RVs independence and their zero-mean property have been omitted. 

For the third term in \eqref{App_eq:caseI_ia}, the sub-case (i), i.e., when $k\neq l\neq 0$, the contribution is zero under the RVs independence and the zero-mean assumption, but when $k = l\neq 0$, i.e., sub-case (ii), the term is reduced to
\begin{align}\label{App_eq:caseI_ic_b}
    cs_m\mathbb{E}\Bigg\{ \sum_{k\in\mathcal{W}} (|A_{\xpol,k}|^2+|A_{\ypol,k}|^2)S_{kl0}\Bigg\} = c s_m \sum_{k\in\mathcal{W}}S_{kk0}.
\end{align}

Putting together all nonzero contributions of Case II, \eqref{App_eq:caseI_ia_b}, \eqref{App_eq:caseI_ib_b}, and \eqref{App_eq:caseI_ic_b}, leads to
\begin{align}\label{App_eq:type_a}
    \mathbb{E}\{\Delta A_{\xpol,0}|A_{\xpol,0}=s_m\}  =  c \sum_{k\in\mathcal{W}} & s_m (2S_{kk0} + S_{k0k})
    \nonumber\\  +  & s_m^* \mathbb{E}\{A^2_{\xpol}\} S_{0kk}.
\end{align}

\subsection{Case III}
In this case, the contributions are
\begin{align}\label{App_eq:caseIII}
    \mathbb{E}\{\Delta A_{\xpol,0}|& A_{\xpol,0}=s_m\} = \nonumber\\  
    &\mathbb{E}\Bigg\{ c\sum_{k \in\mathcal{W}} (A^*_{\xpol,k}s_m+A^*_{\ypol,k}A_{\ypol,0})s_m S_{k00}\Bigg\} \nonumber \\
    & + \mathbb{E}\Bigg\{ c\sum_{l \in\mathcal{W}} (s^*_m A^*_{\xpol,l}+A_{\ypol,0}A_{\ypol,l})s_m S_{0l0}\Bigg\} \nonumber \\
    & + \mathbb{E}\Bigg\{ c\sum_{m \in\mathcal{W}} (|s_m|^2+|A_{\ypol,0}|^2)A_{\xpol,m}S_{00m}\Bigg\},
\end{align}
and all contributions are zero due to the RV's independence and their zero-mean condition.

\subsection{Case IV}
For this case $A_{\xpol/\ypol,k} = A_{\xpol/\ypol,l} = A_{\xpol/\ypol,m} = A_{\xpol/\ypol,0}$, which makes the conditional expectation 
\begin{align}\label{AppB:type_c}
    \mathbb{E}\{\Delta A_{\xpol,0}|A_{\xpol,0}=s_m\} & = c \mathbb{E}\Big\{\left(|s_m|^2 +A^*_{\ypol,0}A_{\ypol,0} \right) s_m \hspace{0.05cm}  S_{000} \Big\} \nonumber \\  
    & = c s_m\left(|s_m|^2 + \mathbb{E}\{|A_{\ypol,0}|^2\}\right)S_{000} \nonumber \\
    & =  c s_m (|s_m|^2+1)S_{000}.
\end{align}

Finally, we put together in \eqref{eq:ex_rp_rvs} all nonzero contributions, i.e., \eqref{App_eq:type_a} and \eqref{AppB:type_c}, and replacing back $c=\jmath \frac{8}{9}\gamma E_s$, results into 
\begin{align}\label{eq:proofT1}
    \tilde\mu_{\xpol,m} &=  s_m  +  \jmath \frac{8}{9}\gamma E_s \Big[s_m(1+ |s_m|^2) S_{000} \nonumber\\ &+\sum_{k\in\mathcal{W}} s_m ( 2S_{kk0}+ S_{k0k}) + s_m^* \mathbb{E}\{A^2_{\xpol}\} S_{0kk}\Big].
\end{align}
The expression in \eqref{eq:proofT1}  can be immediately generalized to both polarization components $A_{\xpol/\ypol}$. In addition, due to the stationary channel assumption \eqref{eq:proofT1} is valid for any time instance, which completes the proof of the conditional mean in Theorem \ref{T:conditional_mean}.

\section{Proof of Theorem \ref{T:NBGD}}\label{app:wirtinger}
 Using the \emph{Wirtinger formalism} \cite[App.~A]{Fisher}, the partial derivatives for a given function $f(z)$ of a complex variable $z=x+\jmath y \in \mathbb{C}$, $x,y \in \mathbb{R}$, with respect to $z$, are defined as  
\begin{equation}\label{eq:wirtinger_derivative}
    \frac{\partial }{\partial z}f \triangleq \frac{1}{2}\left( \frac{\partial }{\partial x}-\jmath \frac{\partial }{\partial y}\right)f,
    \frac{\partial }{\partial z^*}f \triangleq \frac{1}{2}\left( \frac{\partial }{\partial x}+\jmath \frac{\partial }{\partial y}\right)f.
 \end{equation}
For the case of a multi-variable function, $F: \vect{z} \in \mathbb{C}^L \mapsto w \in \mathbb{R} $, all complex derivatives with respect to the complex variables $z_n=x_n+\jmath y_n \in \mathbb{C}$, with  $x_n,y_n \in \mathbb{R}$  $\forall n = 1,\dots, L$ must be calculated.  These derivatives are combined into the \emph{gradient with respect to the Wirtinger derivatives}  as
\begin{align}\label{App_eq:wderivatives}
    \frac{\partial F}{\partial \vect{z}} & \triangleq \Bigg(\frac{\partial F}{\partial z_1},\frac{\partial F}{\partial z_2},\cdots,\frac{\partial F}{\partial z_ L}\Bigg)^\mathrm{T}, \\
    \frac{\partial F}{\partial \vect{z}^*} & \triangleq \Bigg(\frac{\partial F}{\partial z_1^*},\frac{\partial F}{\partial z_2^*},\cdots,\frac{\partial F}{\partial z_ L^*}\Bigg)^\mathrm{T}.
\end{align}

The gradient with respect to the Wirtinger derivatives in \eqref{App_eq:wderivatives} is related to the gradient $\nabla$ required by \eqref{eq:norm_gd} via\footnote{Note the Writinger derivatives with respect to $\vect{z}$ are not needed for the gradient $\nabla$, but are included in \eqref{App_eq:wderivatives} for completeness.}
 \begin{equation} \label{gradient_gradW}
     \nabla F(\vect{z}) = 2 \frac{\partial F}{\partial \vect{z}^*},
 \end{equation}
where
\begin{equation} \label{gradient}
 \nabla \triangleq \begin{pmatrix}
       \frac{\partial }{\partial x_1}+\jmath\frac{\partial }{\partial y_1} \\[0.3em]
       \frac{\partial }{\partial x_2}+\jmath\frac{\partial }{\partial y_2} \\[0.3em]
       \vdots \\[0.3em]
       \frac{\partial }{\partial x_L}+\jmath\frac{\partial }{\partial y_L}           
     \end{pmatrix}.
\end{equation}

 By virtue of the Wirtinger derivatives  $\frac{\partial \hat r_i}{\partial s^*_i}=0$, since each component $r_i$ in \eqref{eq:channel_model_vector} does not depend on $s^*_i$. Equivalently, $\frac{\partial \hat r_i^*}{\partial s_i}=0$. Therefore when taking the derivatives we have 
 \begin{align}\label{eq:gradientK_cost}
    \nabla \text{MSE}(\hat{\vect{s}})& = 2\frac{\partial \text{MSE} }{\partial \hat{\vect{s}}^*} \nonumber\\
    & = \frac{1}{B}\left[(\hat{\vect{r}}-\vect{r})^*\frac{\partial(\hat{\vect{r}}-\vect{r})}{\partial \hat{\vect{s}}^*} + \frac{\partial (\hat{\vect{r}}-\vect{r})^*}{\partial \hat{\vect{s}}^*}(\hat{\vect{r}}-\vect{r})\right]\nonumber \\
    & = -\eta \mathbb{T}^{\dagger}(\hat{\vect{r}}-\vect{r}),
 \end{align}
 with $\eta=\jmath \frac{8}{9}\gamma E_s B$. The batch gradient descent iteration for $\hat{\vect{s}}$ is then given by
 \begin{align}\label{eq:SDG_k}
   \hat{\vect{s}}^{(l+1)} = \hat{\vect{s}}^{(l)}+\alpha \hspace{0.08cm} \eta \mathbb{T}^{\dagger}(\hat{\vect{r}}^{(l)}-\vect{r}),
 \end{align}
 where
 \begin{equation}\label{eq:channel_model_vector_l}
    \hat{\vect{r}}^{(l)} = \vect{a} +\jmath\frac{8}{9} \gamma E_s \mathbb{T} \hat{\vect{s}}^{(l)}.  
 \end{equation}
 Eq. \eqref{eq:channel_model_vector_l} is obtained using \eqref{eq:channel_model_vector} for the estimated kernels in step $l$. 
 Finally, \eqref{eq:SDG_k} is transformed by taking an element-wise normalization. The normalization of \eqref{eq:channel_model_vector_l} completes the NBGD proof in Theorem \eqref{T:NBGD}.

\begin{IEEEbiographynophoto}
{Astrid Barreiro} (Student Member, IEEE) received a B.Sc. in physics (thesis with highest honors) from Universidad del Valle, Cali, Colombia, in 2022. She obtained an M.Sc. degree in Physics from the same university in 2018. Since January 2019 she has been working towards a Ph.D. in electrical engineering at the Eindhoven University of Technology (TU/e), The Netherlands. Her research interest includes the mathematical modeling for optical fiber transmissions in the nonlinear regime, and the design of pragmatic digital signal processing schemes to overcome nonlinear distortions. 
\end{IEEEbiographynophoto}

\begin{IEEEbiographynophoto}
{Gabriele Liga} (Member, IEEE) was born in Palermo, Sicily, Italy, in 1983. He received the B.Sc. degree (Laurea triennale) in telecommunications engineering from Universita' degli Studi di Palermo in 2005, and the M.Sc. degree in telecommunications engineering (Laurea specialistica) from Politecnico di Milano in 2011. In 2017, he obtained the Ph.D. degree in optical communications from the Optical Networks Group, Electronics and Electrical Engineering Department, University College London, United Kingdom. From 2017 to 2018, he worked as a Postdoctoral Research Associate with the Optical Networks Group, focussing on digital signal processing and nonlinearity compensation techniques for optical fiber transmission. In 2018, he was awarded a Marie Sklodowska-Curie EurotechPostdoc programme fellowship to work on signal shaping tailored to the nonlinear optical fiber channel within the Signal Processing Systems (SPS) Group, Department of Electrical Engineering, Eindhoven University of Technology (TU/e), Eindhoven, The Netherlands.,He currently serves as a Reviewer for several scientific journals in the area of communications and photonics, such as IEEE Journal of Lightwave Technology, IEEE Transactions on Information Theory, IEEE Transactions on Communications, OSA Optics Express, and IEEE Photonics Technology Letters. His research interests embrace the areas of digital communications, mathematical modeling and information theory applied to fiber-optic.
\end{IEEEbiographynophoto}

\begin{IEEEbiographynophoto}
{Alex Alvarado} (S'06--M'11--SM'15) was born in Quellón, on the island of Chiloé, Chile. He received his Electronics Engineer degree (Ingeniero Civil Electrónico) and his M.Sc. degree (Magíster en Ciencias de la Ingeniería Electrónica) from Universidad Técnica Federico Santa María, Valparaíso, Chile, in 2003 and 2005, respectively. He obtained the degree of Licentiate of Engineering (Teknologie Licentiatexamen) in 2008 and his PhD degree in 2011, both of them from Chalmers University of Technology, Gothenburg, Sweden. 
 
Dr. Alvarado is associate professor at the Signal Processing Systems (SPS) Group, Department of Electrical Engineering, Eindhoven University of Technology (TU/e), The Netherlands. During 2018-2022 he was a member of the TU/e Young Academy of Engineering. During 2014—2016, he was a Senior Research Associate at the Optical Networks Group, University College London, United Kingdom. In 2012—2014 Dr. Alvarado was a Marie Curie Intra-European Fellow at the University of Cambridge, United Kingdom, and during 2011—2012 he was a Newton International Fellow at the same institution. Dr. Alvarado's research has been funded in part by the Netherlands Organisation for Scientific Research (NWO) via a VIDI grant, as well as by the European Research Council (ERC) via an ERC Starting Grant. 
 
Dr. Alvarado's research has received multiple awards, including Best Paper Awards at the 2018 Asia Communications and Photonics Conference and at the 2019 OptoElectronics and Communications Conference, and Best Poster Awards at the 2009 IEEE Information Theory Workshop and at the 2013 IEEE Communication Theory Workshop. He is also recipient of the 2015 IEEE Transactions on Communications Exemplary Reviewer Award, and the 2015 Journal of Lightwave Technology Best Paper Award, honoring the most influential, highest-cited original paper published in the journal in 2015. Dr. Alvarado is a senior member of the IEEE and served as an associate editor for IEEE Transactions on Communications (Optical Coded Modulation and Information Theory) during 2016-2018. During 2018–2020, he served in the OFC subcommittee Digital and Electronic Subsystems (S4). He also served in the ECOC subcommittee Theory of Optical Communications during 2019–2022. His general research interests are in the  areas of digital communications, coding, and information theory. 
\end{IEEEbiographynophoto}
\end{document}